\documentclass[11pt]{article} 
\usepackage{amsfonts,amssymb,amsmath,amsthm,graphicx,hyperref,mathrsfs,bm,bbm}
\usepackage[title, titletoc]{appendix}
\usepackage[onehalfspacing]{setspace}
\usepackage{caption,subcaption,lscape}
\hypersetup{pdfborder = {0 0 0},colorlinks=true,linkcolor=blue,urlcolor=blue,citecolor=blue}
\usepackage{natbib}
\usepackage[shortlabels]{enumitem}

\usepackage{xcolor}

\usepackage[top=1in,bottom=1in,left=1in,right=1in]{geometry}

\newtheorem{thm}{Theorem}
\newtheorem{coro}{Corollary}
\newtheorem{lem}{Lemma}
\newtheorem{assumption}{Assumption}

\numberwithin{equation}{section}

\theoremstyle{definition}
\newtheorem{remark_tmp}{Remark}[section]
\newenvironment{remark}
	{ \begin{remark_tmp} 	}
	{ 
		\medskip\hfill{\LARGE$\lrcorner$}
		\end{remark_tmp} 
	}

\allowdisplaybreaks[1]

\renewcommand{\P}{\mathbb{P}}
\newcommand{\E}{\mathbb{E}}
\newcommand{\V}{\mathbb{V}}
\newcommand{\R}{\mathbb{R}}
\newcommand{\N}{\mathbb{N}}
\newcommand{\Z}{\mathbb{Z}}
\newcommand{\1}{\mathbbm{1}}

\def\m{\mathcal}

\def\ls{\lesssim}

\def\d{\mathrm{d}}

\newcommand{\Center}{\theta}
\newcommand{\Scale}{\vartheta}
\newcommand{\Lincom}{\mathsf{v}}
\renewcommand{\iota}{\mathsf{i}}

\newcommand{\AL}{\mathtt{AL}}
\newcommand{\SB}{\mathtt{SB}}

\begin{document}

\title{\vspace{-0.0in} Higher-Order Refinements of Small Bandwidth Asymptotics for Density-Weighted Average Derivative Estimators\thanks{Prepared for the Conference in Honor of James L. Powell at UC Berkeley, March 25--26, 2022. We thank the conference participants for their comments, as well as two reviewers, an Associate Editor, and co-Editor Elie Tamer for their feedback. Cattaneo gratefully acknowledges financial support from the National Science Foundation through grants SES-1947805 and DMS-2210561, Jansson gratefully acknowledges financial support from the National Science Foundation through grant SES-1947662, and Masini gratefully acknowledges financial support from the National Science Foundation through grant DMS-2210561.}
\bigskip }
\author{Matias D. Cattaneo\thanks{Department of Operations Research and Financial Engineering, Princeton University.} \and
	    Max H. Farrell\thanks{Department of Economics, UC Santa Barbara.} \and
	    Michael Jansson\thanks{Department of Economics, UC Berkeley.} \and
	    Ricardo P. Masini\thanks{Department of Statistics, UC Davis.}}
\maketitle

\begin{abstract}
    The density weighted average derivative (DWAD) of a regression function is a canonical parameter of interest in economics. Classical first-order large sample distribution theory for kernel-based DWAD estimators relies on tuning parameter restrictions and model assumptions that imply an asymptotic linear representation of the point estimator. These conditions can be restrictive, and the resulting distributional approximation may not be representative of the actual sampling distribution of the statistic of interest. In particular, the approximation is not robust to bandwidth choice. Small bandwidth asymptotics offers an alternative, more general distributional approximation for kernel-based DWAD estimators that allows for, but does not require, asymptotic linearity. The resulting inference procedures based on small bandwidth asymptotics were found to exhibit superior finite sample performance in simulations, but no formal theory justifying that empirical success is available in the literature. Employing Edgeworth expansions, this paper shows that small bandwidth asymptotic approximations lead to inference procedures with higher-order distributional properties that are demonstrably superior to those of procedures based on asymptotic linear approximations.
\end{abstract}

\thispagestyle{empty}
\clearpage

\setcounter{page}{1}
\pagestyle{plain}

\section{Introduction}

Identification, estimation, and inference in the context of semiparametric models has a long tradition in econometrics \citep{Powell_1994_Handbook}. Canonical two-step semiparametric estimands are finite dimensional functionals of some other unknown infinite dimensional parameters in the model (e.g., a density or regression function). A leading example of such a finite dimensional estimand is the density weighted average derivative (DWAD) of a regression function. This paper seeks to honor the many contributions of Jim Powell to semiparametric theory in econometrics by juxtaposing the higher-order distributional properties of \citeauthor{Powell-Stock-Stoker_1989_ECMA}'s \citeyearpar{Powell-Stock-Stoker_1989_ECMA} two-step kernel-based DWAD estimator under two alternative large sample approximation regimes: one based on the classical asymptotic linear representation, and the other based on a more general quadratic distributional approximation known as \textit{small bandwidth} asymptotics.\footnote{Jim Powell's contributions to semiparametric theory are numerous. \citet{Honore-Powell_1994_JOE}, \citet{Powell-Stoker_1996_JoE}, \citet{Blundell-Powell_2004_RESTUD}, \citet{AradillasLopez-Honore-Powell_2007_IER}, \citet{Ahn-Ichimura-Powell-Ruud_2018_JBES}, and \citet{Graham-Niu-Powell_2023_JOE} are some of the most closely connected to the our work: these papers employ U-statistics methods for two-step kernel-based estimators similar to those considered herein. See \citet{Powell_2017_JEP} for more discussion and references.}

In a landmark contribution, \citet{Powell-Stock-Stoker_1989_ECMA} proposed a kernel-based DWAD estimator and obtained first-order, asymptotically linear distribution theory employing ideas from the U-statistics literature in statistics to develop valid inference procedures in large samples. This work sparked a wealth of subsequent developments in the econometrics literature: \citet{Robinson_1995_ECMA} obtained Berry-Esseen bounds, \citet{Powell-Stoker_1996_JoE} considered mean square error expansions, \citet{Nishiyama-Robinson_2000_ECMA,Nishiyama-Robinson_2001_ChBook,Nishiyama-Robinson_2005_ECMA} developed Edgeworth expansions, and \citet{Newey-Hsieh-Robins_2004_Ecma} investigated bias properties, just to mention a few contributions. The two-step semiparametric estimator in this literature employs a preliminary kernel-based estimator of a density function, which requires choosing two main tuning parameters: a bandwidth and a kernel function. The ``optimal'' choices for these tuning parameters depend on the goal of interest (e.g., point estimation vs. inference), as well as on the features of the underlying data generating process (e.g., smoothness of the unknown density and dimension of the covariates).  

Classical first-order distribution theory for kernel-based DWAD estimators has focused on cases where tuning parameter restrictions and model assumptions imply an asymptotic linear representation of the two-step semiparametric point estimator \citep[see][for overviews]{Bickel-Klaassen-Ritov-Wellner_1993_Book,Newey-McFadden_1994_Handbook,Ichimura-Todd_2007_Handbook}. That is, the two-step estimator is approximated by a sample average based on an influence function. This approach can be used to construct semiparametrically efficient inference procedures, but requires potentially high smoothness levels of the underlying unknown functions, thereby forcing the use of higher-order kernels or other debiasing techniques. Further, the implied distributional approximation may not be ``robust'' to tuning parameter choices and/or model features. More specifically, the limiting distribution emerging from the asymptotic linear representation of the centered and scaled point estimator is invariant to the way that the preliminary nonparametric estimators are constructed. At its core, an asymptotic linear approximation assumes away the contribution of additional terms forming the statistic of interest, despite the fact that these terms do contribute to the sampling variability of the two-step semiparametric estimator and, more importantly, do reflect the impact of tuning parameter choices in finite samples.

\citet{Cattaneo-Crump-Jansson_2014a_ET} proposed an alternative distributional approximation for kernel-based DWAD estimators that allows for, but does not require, asymptotic linearity. The idea is to capture the joint contribution to the sampling distribution of both linear and quadratic terms forming the kernel-based DWAD estimator, because the quadratic term explicitly captures the effect of the choice of bandwidth and kernel function. To operationalize this idea, \citet{Cattaneo-Crump-Jansson_2014a_ET} introduced an asymptotic experiment where the bandwidth sequence is allowed (but not required) to vanish at a speed that would render the classical asymptotic linear representation invalid because the quadratic term becomes first order even in large samples, which they termed \textit{small bandwidth} asymptotics. This framework was carefully developed to obtain a distributional approximation that explicitly depends on both linear and quadratic terms, thereby forcing a more careful analysis of how the nonparametric first stage contributes to the sampling distribution of the statistic.

Inference methods based on small bandwidth asymptotics for kernel-based DWAD estimators were found to perform well in simulations \citep{Cattaneo-Crump-Jansson_2010_JASA,Cattaneo-Crump-Jansson_2014a_ET,Cattaneo-Crump-Jansson_2014b_ET}, but no formal justification for this finite sample success is available in the literature. Methodologically, this alternative distributional approximation leads to a new way of conducting inference (e.g., constructing confidence interval estimators) because the original standard error formula proposed by \citet{Powell-Stock-Stoker_1989_ECMA} must be modified to make the asymptotic approximation valid across the full range of allowable bandwidths (including the region where asymptotic linearity fails). Theoretically, however, the empirical success of small bandwidth asymptotics could come from two distinct sources: (i) it could deliver a better distributional approximation to the sampling distribution of the point estimator; or (ii) it could deliver a better distributional approximation to the sampling distribution of the studentized t-statistic because the standard error formula is modified.

Employing Edgeworth expansions \citep{Bhattacharya-Rao_1976_Book,Hall_1992_Book}, this paper shows that the higher-order distributional properties of inference procedures motivated by the small bandwidth asymptotics approximation framework are demonstrably superior to those of procedures motivated by asymptotic linear approximations. We study both standardized and studentized estimators and show that those emerging from the small bandwidth regime offer higher-order corrections, as measured by the second cumulant underlying their Edgeworth expansions. An immediate implication of our results is that the small bandwidth asymptotic framework simultaneously enjoys two advantages: delivering a better distributional approximation (Theorem \ref{thm:EE-standard}, standardized t-statistic) and leading to a better standard error construction (Theorem \ref{thm:EE-student}, studentized t-statistic). Therefore, our results have theoretical and practical implications for empirical work in economics, in addition to providing a theory-based explanation for prior simulation findings documenting better numerical performance of inference procedures motivated by small bandwidth asymptotics relative to those motivated by classical asymptotically linear distributional approximations.

The closest antecedent to our work is \citet{Nishiyama-Robinson_2000_ECMA,Nishiyama-Robinson_2001_ChBook}, who also studied Edgeworth expansions for kernel-based DWAD estimators. Their expansions, however, were motivated by the asymptotic linear approximation of the point estimator, and hence cannot be used to compare and contrast to the distributional approximation emerging from the alternative small bandwidth asymptotic regime. Therefore, from a technical perspective, this paper also offers novel Edgeworth expansions that allow for different standardization and studentization schemes, thereby allowing us to plug-and-play when juxtaposing the two asymptotic approximation frameworks. More specifically, Theorem \ref{thm:EE-standard} below concerns a generic standardized t-statistic and is proven based on Theorem \ref{App A: Second-Order U-Statistics} in the appendix, which may be of independent technical interest due to is generality. In contrast, Theorem \ref{thm:EE-student} below concerns a more specialized class of studentized t-statistic because establishing valid Edgeworth expansions is considerably harder when dealing with studentization.

The idea of employing more general asymptotic approximation frameworks that do not enforce asymptotic linearity for two-step semiparametric estimators has also featured in other contexts: (i) semi-linear series-based, many covariates, and many instruments estimation \citep{Cattaneo-Jansson-Newey_2018_ET,Cattaneo-Jansson-Newey_2018_JASA}, (ii) non-linear two-step semiparametric estimation \citep{Cattaneo-Crump-Jansson_2013_JASA,Cattaneo-Jansson_2018_ECMA,Cattaneo-Jansson-Ma_2019_RESTUD,Cattaneo-Jansson_2022_ET}, and (iii) network estimation \citep{Matsushita-Otsu_2021_Biometrika}. While our theoretical developments and results focus specifically on the case of kernel-based DWAD estimation, our main conceptual conclusions can be extrapolated to those settings as well. The main takeaway is that employing alternative asymptotic frameworks can deliver improved inference with smaller higher-order distributional approximation errors, thereby offering more robust inference procedures in finite samples. Furthermore, our theoretical and methodological results can also be leveraged to study the higher-order distributional properties of bootstrap-based methods for inference. Although a complete theoretical analysis is beyond the scope of this paper, we provide further discussion about the bootstrap in Section \ref{sec: Higher-Order Distribution Theory}.

The paper continues as follows. Section \ref{sec: Setup} introduces the setup and main assumptions. Section \ref{sec: First-Order Distribution Theory} reviews the classical first-order distributional approximation based on asymptotic linearity and the more general small bandwidth distributional approximation, along with their corresponding choices of standard error formulas. Section \ref{sec: Higher-Order Distribution Theory} presents the main results of our paper. Section \ref{sec: Conclusion} concludes. The appendix is organized in three parts: Appendix \ref{App A: Second-Order U-Statistics} provides a self-contained generic Edgeworth expansion for second-order U-statistics, which may be of independent technical interest, Appendix \ref{App B: Proof of Theorem EE-standard} gives the proof of Theorem \ref{thm:EE-standard}, and Appendix \ref{App C: Proof of Theorem EE-student} gives the proof of Theorem \ref{thm:EE-student}.

\section{Setup and Assumptions}\label{sec: Setup}

Suppose $Z_i=(Y_i, X_i')'$, $i = 1,\dots,n$, is a random sample from the distribution of the random vector $Z=(Y, X')'$, where $Y$ is an outcome variable and $X$ takes values on $\R^d$ with Lebesgue density $f$. We consider
\[\theta := \E[f(X)\dot{g}(X)], \qquad g(X) := \E[Y|X],\]
the DWAD of the regression function $g$, where, for any (differentiable) function $a$, $\dot{a}(x)$ denotes $\partial a(x) / \partial x$, and where existence of $\theta$ is implied by parts (b) and (c)) of the following assumption, which collects the regularity conditions under which our subsequent analysis will proceed.

\begin{assumption}\label{A:DGP}
    For some $S \in \N$, the following are satisfied:
    \begin{enumerate}[(a)]
        \item $\E[|Y|^3]<\infty$;
        \item $f$ is $(S+1)$ times differentiable, and $f$ and its first $(S+1)$ derivatives are bounded;
        \item $g$ is $(S+1)$ times differentiable and its first three derivatives are bounded;
        \item $e$ and its first $(S+1)$ derivatives are bounded, where $e(X) := f(X)g(X)$;
        \item $\E[\V(Y|X)f(X)]>0$ and $\Sigma := \V[\psi(Z)]$ is positive definite, where $\psi(Z) := 2[\dot{e}(X) -Y\dot{f}(X) - \theta]$ and where $\V[\cdot]$ denotes the variance;
        \item $v$ is twice differentiable, its first two derivatives are bounded, and $v\dot{f}$ and $\E[|Y|^3|X] f(X)$ are bounded, where $v(X) := \E[Y^2|X]$;
        \item $\limsup_{||x|| \to \infty} [1 + v(x)]f(x) = 0$, where $||\cdot||$ is the Euclidean norm; and
        \item Cram\'er Condition: For every $\Lincom\in\R^d,$ 
        \begin{equation*}
            \limsup\limits_{|t|\to\infty}\left|\E\left[\exp\left(\iota t\psi_\Lincom(Z)\right)\right]\right|<1,
        \end{equation*}
        where $\psi_\Lincom(Z) := \Lincom'\psi(Z)$ and $\iota^2 := -1$.
    \end{enumerate}
\end{assumption}

Under Assumption \ref{A:DGP} and using integration by parts, the DWAD vector can be expressed as
\begin{equation*}
    \theta = -2\E[Y \dot{f}(X)],
\end{equation*}
which motivates the celebrated plug-in analog estimator of \citet{Powell-Stock-Stoker_1989_ECMA} given by
\begin{equation*}
    \widehat{\theta} = -2 n^{-1} \sum_{1 \leq i \leq n} Y_i \frac{\partial}{\partial x}\widehat{f_i}(X_i), \qquad \widehat{f_i}(x) = (n-1)^{-1} \sum_{\substack{1 \leq j \leq n \\ j \neq i}}  \frac{1}{h^d} K\left(\frac{X_j-x}{h}\right),
\end{equation*}
where $\widehat{f_i}(\cdot)$ is a ``leave-one-out'' kernel density estimator employing a symmetric and differentiable kernel function $K:\mathbb{R}^{d}\rightarrow \mathbb{R}$ and a positive vanishing (bandwidth) sequence $h$.

The estimator $\widehat{\theta}$ can be expressed as a second-order U-statistic with an $n$-varying kernel:
\begin{equation}\label{eq:estimator}
    \widehat{\theta} = \binom{n}{2}^{-1}\sum_{1 \leq i < j \leq n} U_{ij}, \qquad
    U_{ij} := -h^{-d-1}\dot{K}\left(\frac{X_i-X_j}{h}\right)(Y_i-Y_j).
\end{equation}
Our analysis of $\widehat{\theta}$ is based on this representation and proceeds under the following assumption about the kernel function.\footnote{In Assumption \ref{A:kernel} (c) and elsewhere, we employ standard multi-index notation: For $a := (a_1,\dots,a_d)' \in \Z_+^d$, we have (i) $[a] := a_1+\dots +a_d$, (ii) $a! := a_1!\dots a_d!$, (iii) $x^a := x_1^{a_1}\dots x_d^{a_d}$ for $x := (x_1,\dots,x_d)' \in \R^d$, and (iv) $\partial^a q(x) / \partial x^a := \partial^{[a]} q(x) / (\partial x_1^{a_1}\dots \partial x_d^{a_d})$ for (sufficiently smooth) $q:\R^d\to\R$.}

\begin{assumption}\label{A:kernel}
    For some $P \geq 2$ and some $\{\mu_a: a \in \Z^d_+, [a] = P\} \subseteq \R$, the following are satisfied:
    \begin{enumerate}[(a)]
        \item $K$ is even, differentiable, and $\dot{K}$ is bounded;
        \item $\int_{\R^d} \dot{K}(u)\dot{K}(u)' \mathrm{d}u $ is positive definite; and
        \item $\int_{\R^d} |K(u)|(1+\|u\|^P) \mathrm{d}u + \int_{\R^d} \|\dot{K}(u)\|(1+\|u\|^2) \mathrm{d}u<\infty$ and
        \begin{equation*}
            \int_{\R^d} u^a K(u) \mathrm{d}u =
            \begin{cases}
                1, &\text{if } [a] =  0,\\
                0, &\text{if } 0< [a] < P\\
                \mu_a, &\text{if } [a] = P,
            \end{cases}
        \end{equation*}
        where $a \in \Z^d_+$ is a multi-index.
    \end{enumerate}
\end{assumption}

\section{First-Order Distribution Theory}\label{sec: First-Order Distribution Theory}

Before presenting our main results concerning the higher-order distributional properties of different statistics based on $\widehat{\theta}$, we review conventional and alternative first-order asymptotic distributional approximations, as well as the distinct variance estimation methods emerging from each of those approximation frameworks. Limits are taken as $h\to0$ and $n\to\infty$ unless otherwise noted, $\to_\P$ denotes convergence in probability, and $\rightsquigarrow$ denotes convergence in law.

\subsection{Distributional Approximation}

Under appropriate restrictions on $h$ and $K$, the estimator $\widehat{\theta}$ is asymptotically linear with influence function $\psi$ and asymptotic variance $\Sigma$. More precisely, \citet{Powell-Stock-Stoker_1989_ECMA} showed that if Assumptions \ref{A:DGP} and \ref{A:kernel} hold and if $nh^{2(P \land S)}\to 0$ and $nh^{d+2}\to \infty $ (where $a \land b$ denotes $\min(a,b)$), then
\begin{equation}\label{eq:AL-Distribution}
    \sqrt{n}(\widehat{\theta} - \theta) = n^{-1/2} \sum_{1 \leq i \leq n} \psi(Z_i) + o_\P(1)
    \rightsquigarrow \m{N}(0, \Sigma).
\end{equation}
A proof of \eqref{eq:AL-Distribution} can be based on the $U$-statistic representation in \eqref{eq:estimator} and its Hoeffding decomposition $\widehat{\theta} = \E[U_{ij}] + \bar{L} + \bar{Q}$, where $\bar{L}$ and $\bar{Q}$ are mean zero random vectors given by
\begin{equation*}
    \bar{L} := n^{-1} \sum_{1 \leq i \leq n} L_i, \qquad L_i := 2(\E[U_{ij}|Z_i] - \E[U_{ij}]),
\end{equation*}
and
\begin{equation*}
    \bar{Q} := \binom{n}{2}^{-1}\sum_{1 \leq i < j \leq n} Q_{ij}, \qquad Q_{ij} := U_{ij} - \E[U_{ij}|Z_i] - \E[U_{ij}|Z_j] + \E[U_{ij}],
\end{equation*}
respectively: because $\E[U_{ij}] = \theta + O(h^{P \land S})$ and $\V[\bar{Q}]=O(n^{-2} h^{-d-2})$, we have
\begin{equation*}
    \sqrt{n}(\widehat{\theta} - \theta) = n^{-1/2} \sum_{i=1}^n L_i + O_\P\left(\sqrt{n}h^{P \land S} + \frac{1}{\sqrt{n h^{d+2}}}\right),
\end{equation*}
from which the result \eqref{eq:AL-Distribution} follows upon noting that $\V[L_i - \psi(Z_i)] = O(h^{P \land S})$. Using Edgeworth expansions, \cite{Nishiyama-Robinson_2000_ECMA,Nishiyama-Robinson_2001_ChBook} studied the quality of the distributional approximation implied by \eqref{eq:AL-Distribution}; their result is contained as a special case of our Theorem \ref{thm:EE-standard}.

The Hoeffding decomposition and subsequent analysis of each of its terms shows that the estimator admits a bilinear form representation in general, which then is reduced to a sample average approximation by assuming a bandwidth sequence and kernel shape that makes both the misspecification error (smoothing bias) and the variability introduced by $\bar{Q}$ (a ``quadratic'' term) negligible in large samples. As a result, provided that such tuning parameter choices are feasible, the estimator will be asymptotically linear.

Asymptotic linearity of a semiparametric estimator has several distinct features that may be considered attractive from a theoretical point of view. In particular, it is a necessary condition for semiparametric efficiency, and it leads to a limiting distribution that is invariant to the choice of the first-step nonparametric estimator entering the two-step semiparametric procedure \citep{Newey_1994_ECMA}. However, insisting on asymptotic linearity may also have its drawbacks because it requires several potentially strong assumptions, and because it leads to a large sample theory that may not accurately represent the finite sample behavior of the statistic. In the case of $\widehat{\theta}$, asymptotic linearity requires $P>2$ unless $d=1$; that is, the use of higher-order kernels or similar debiasing techniques \citep[see, e.g.,][and references therein]{Chernozhukov-etal_2022_ECMA} is necessary in order to achieve asymptotic linearity. In addition, asymptotic linearity leads to a limiting experiment which is invariant to the particular choices of smoothing ($K$) and bandwidth ($h$) tuning parameters involved in the construction of the estimator. As a result, large sample distribution theory based on (or implying) asymptotic linearity is silent with respect to the impact that tuning parameter choices may have on the finite sample behavior of the two-step semiparametric statistic.

To address the aforementioned limitations of distribution theory based on asymptotic linearity, \citet{Cattaneo-Crump-Jansson_2014a_ET} proposed a more general distributional approximation for kernel-based DWAD estimators that accommodates, but does not enforce, asymptotic linearity. The idea is to characterize the joint asymptotic distributional features of both the linear ($\bar{L}$) and quadratic ($\bar{Q}$) terms, and in the process develop a more general first-order asymptotic theory that allows for weaker assumptions than those imposed in the classical asymptotically linear distribution theory. Formally, if Assumptions \ref{A:DGP} and \ref{A:kernel} hold,
and if $(nh^{d+2} \land 1) nh^{2(P \land S) }\to 0$ and $n^{2}h^{d}\to \infty$, then
\begin{equation}\label{eq:SB-Distribution}
    \V[\widehat{\theta}]^{-1/2}(\widehat{\theta}-\theta)\rightsquigarrow\m{N}(0,I),
\end{equation}
where $\V[\widehat{\theta}] = \V[\bar{L}] + \V[\bar{Q}]$ with
\begin{equation*}
    \V[\bar{L}] = n^{-1} \left[ \Sigma + O\left(h^{P \land S}\right) \right]
\end{equation*}
and
\begin{equation*}
    \V[\bar{Q}] = \binom{n}{2}^{-1} h^{-d-2} \left[ \Delta + O(h^2) \right], \qquad \Delta := 2\E[\V(Y|X)f(X)] \int_{\mathbb{R}^{d}}\dot{K}(u) \dot{K}(u)'\text{d}u.
\end{equation*}

This more general distributional approximation was developed explicitly in an attempt to better characterize the finite sample behavior of $\widehat{\theta}$. The result in \eqref{eq:SB-Distribution} shows that the conditions on the bandwidth sequence may be
considerably weakened without invalidating the limiting Gaussian distribution, although the asymptotic variance formula changes. Importantly, if $nh^{d+2} $ is bounded then $\widehat{\theta}$ is no longer asymptotically linear and its limiting distribution will cease to be invariant with respect to the underlying preliminary nonparametric estimator. In particular, if $nh^{d+2}\to c >0$ then $\widehat{\theta}$ is root-$n$ consistent, but not asymptotically linear. The bias of the estimator is also controlled in a different way because the bandwidth is allowed to be ``smaller'' than usual, which may remove the need for higher-order kernels. Interestingly, \eqref{eq:SB-Distribution} allows for the point estimator to not even be consistent for $\theta$, which occurs for sufficiently small bandwidth sequences.

Beyond the aforementioned technical considerations, the result in \eqref{eq:SB-Distribution} can conceptually be interpreted as a more refined first-order distributional approximation for $\widehat{\theta}$, which by relying on a quadratic approximation (i.e., accounting for the contributions of both $\bar{L}$ and $\bar{Q}$) is expected to offer a ``better'' distributional approximation than approximations relying on asymptotic linearity (i.e., accounting only for the contribution of $\bar{L}$). The idea of standardizing a U-statistic by the joint variance of the linear and quadratic terms underlying its Hoeffding decomposition can be traced back to the original paper of \citet[p. 307]{Hoeffding_1948_AMS}. Simulation evidence reported in \citet{Cattaneo-Crump-Jansson_2010_JASA,Cattaneo-Crump-Jansson_2014a_ET,Cattaneo-Crump-Jansson_2014b_ET} corroborated those conceptual interpretations numerically, but no formal justification is available in the literature. Theorem \ref{thm:EE-standard} below will offer the first theoretical result in the literature highlighting specific robustness features of the distributional approximation in \eqref{eq:SB-Distribution} by showing that such approximation has a demonstrably smaller higher-order distributional approximation error.  

\subsection{Variance Estimation}

Motivated by the asymptotic linearity result \eqref{eq:AL-Distribution}, \citet{Powell-Stock-Stoker_1989_ECMA} also proposed the ``plug-in'' variance estimator
\begin{equation*}
    \widehat{\Sigma} := n^{-1} \sum_{1 \leq i \leq n} \widehat{L}_i\widehat{L}_i', \qquad
    \widehat{L}_i := 2\Big[ (n-1)^{-1} \sum_{\substack{1 \leq j \leq n \\ j \neq i}} U_{ij} -\widehat{\theta} \Big],
\end{equation*}
and proved its consistency (i.e., $\widehat{\Sigma}\to_\P \Sigma$) under the same bandwidth sequences required for asymptotic linearity (i.e., assuming $nh^{2(P \land S)}\to 0$ and $nh^{d+2}\to \infty $). Combining this consistency result with \eqref{eq:AL-Distribution}, we obtain the following result about a studentized version of $\widehat{\theta}$:
\begin{equation}\label{eq:T-stat-AL}
    \widehat{V}_\AL^{-1/2} (\widehat{\theta}-\theta)\rightsquigarrow\m{N}(0,I), \qquad \widehat{V}_\AL := n^{-1} \widehat{\Sigma}.
\end{equation}
Using Edgeworth expansions, \cite{Nishiyama-Robinson_2000_ECMA,Nishiyama-Robinson_2001_ChBook} studied the quality of the distributional approximation implied by \eqref{eq:T-stat-AL}; their result is contained as a special case of our Theorem \ref{thm:EE-student}.

Complementing the small bandwidth asymptotic representation \eqref{eq:SB-Distribution}, \citet{Cattaneo-Crump-Jansson_2014a_ET} showed that
\begin{equation*}
    \widehat{V}_\AL = n^{-1} [\Sigma + o_\P(1) ] + 2\binom{n}{2}^{-1} h^{-d-2}[\Delta+o_\P(1)],
\end{equation*}
which implies among other things that the consistency result $\widehat{\Sigma}\to_\P \Sigma$ is valid only if $nh^{d+2}\to \infty$; otherwise, $\widehat{\Sigma}$ is in general asymptotically upwards biased relative to $\V[\widehat{\theta}]$ in \eqref{eq:SB-Distribution}. Because $\widehat{\Sigma}$ is asymptotically equivalent to the jackknife variance estimator of $\widehat{\theta}$, \citet{Cattaneo-Crump-Jansson_2014b_ET} also noted that the asymptotic bias of $\widehat{\Sigma}$ is a consequence of a more generic phenomena underlying jackknife variance estimators studied in \citet{Efron-Stein_1981_AOS}. See also \citet{Matsushita-Otsu_2021_Biometrika} for related discussion. 

To conduct asymptotically valid inference under the more general small bandwidth asymptotic regime, \citet{Cattaneo-Crump-Jansson_2014a_ET} proposed several ``debiased'' variance estimators, including
\begin{equation*}
    \widehat{V}_\SB := n^{-1} \widehat{\Sigma} - \binom{n}{2}^{-1}h^{-d-2}\widehat{\Delta},\qquad \widehat{\Delta} := h^{d+2} \binom{n}{2}^{-1} \sum_{1 \leq i < j \leq n} U_{ij} U_{ij}',
\end{equation*}
and showed that $\widehat{\Delta}\to_\P \Delta$ under the same bandwidth sequences required for \eqref{eq:SB-Distribution} to hold (i.e., assuming $nh^{2(P \land S)}\to 0$ and $n^2 h^d \to \infty$). The estimator $\widehat{\Delta}$ is asymptotically equivalent to the debiasing procedure proposed in \citet{Efron-Stein_1981_AOS}. By design, the result
\begin{equation}\label{eq:T-stat-SB}
    \widehat{V}_\SB^{-1/2}(\widehat{\theta}-\theta)\rightsquigarrow\m{N}(0,I)
\end{equation}
holds under more general conditions than those required for \eqref{eq:T-stat-AL}, suggesting that inference procedures based on $\widehat{V}_\SB$ are more ``robust'' than procedures based on $\widehat{V}_\AL$.

Conceptually, robustness manifests itself in two distinct ways. First, the underlying Gaussian distributional approximation holds under weaker bandwidth restrictions, a property achieved in part by employing a standardization factor depending explicitly on tuning parameter choices. Second, the new variance estimator $\widehat{V}_\SB$ is obtained from the more general small bandwidth approximation and explicitly accounts for the contribution of terms regarded as higher-order under asymptotic linearity.

While not reproduced here to conserve space, the in-depth Monte Carlo evidence reported in \citet{Cattaneo-Crump-Jansson_2010_JASA,Cattaneo-Crump-Jansson_2014a_ET,Cattaneo-Crump-Jansson_2014b_ET} also showed that employing inference procedures based on \eqref{eq:T-stat-SB} lead to remarkable improvements in terms of ``robustness'' to bandwidth choice and other tuning inputs, when compared to classical asymptotically linear inference procedures based on \eqref{eq:T-stat-AL}. Theorem \ref{thm:EE-student} below will show formally that the distributional approximation \eqref{eq:T-stat-SB} has demonstrably smaller higher-order errors than the distributional approximation \eqref{eq:T-stat-AL}, thereby providing a theory-based explanation for the empirical success of feasible inference procedures developed under the small bandwidth approximation framework.

\section{Higher-Order Distribution Theory}\label{sec: Higher-Order Distribution Theory}

Letting $\Lincom\in\R^d$ be a fixed vector and defining $\widehat{\theta}_\Lincom:=\Lincom'\widehat{\theta}$, this section presents Edgeworth expansions for standardized and studentized statistics based on $\widehat{\theta}_\Lincom$. Section \ref{subsec: Standardized Statistics} studies standardized statistics of the form $(\widehat{\theta}_\Lincom - \theta_\Lincom)/\Scale_\Lincom$, where $\theta_\Lincom:=\Lincom'\theta$ and $\Scale_\Lincom$ is an approximate standard deviation of $\widehat{\theta}_\Lincom$; that is, $\Scale_\Lincom$ is positive, non-random, and such that $(\widehat{\theta}_\Lincom - \theta_\Lincom)/\Scale_\Lincom$ is asymptotically standard normal. The main purpose of studying standardized statistics is to allow us to compare the quality of the distributional approximations \eqref{eq:AL-Distribution} and \eqref{eq:SB-Distribution} based on asymptotic linearity and small bandwidth asymptotics, respectively. Section \ref{subsec: Studentized Statistics} then studies studentized statistics of the form $(\widehat{\theta}_\Lincom - \theta_\Lincom)/\widehat{\Scale}_\Lincom$, where $\widehat{\Scale}_\Lincom^2$ is (random and) equal to either $\Lincom'\widehat{V}_\AL\Lincom$ or $\Lincom'\widehat{V}_\SB\Lincom$. The main purpose of studying studentized statistics is to allow us to investigate the impact of variance estimation on the quality of the distributional approximations \eqref{eq:T-stat-AL} and \eqref{eq:T-stat-SB} based on asymptotic linearity and small bandwidth asymptotics, respectively.

\citet{Nishiyama-Robinson_2000_ECMA,Nishiyama-Robinson_2001_ChBook} obtained valid Edgeworth expansions for the distribution of the standardized and studentized statistics employing $\Scale_\Lincom^2 = \Lincom'\Sigma\Lincom/n$ and $\widehat{\Scale}_\Lincom^2 = \Lincom'\widehat{\Sigma}\Lincom/n$, respectively. Those results were obtained under assumptions implying asymptotic linearity. Although our main interest is in standardization and studentization schemes whose (first-order) validity does not require asymptotic linearity, we retain the assumption of asymptotic linearity to ensure a fair comparison; that is, our results are derived under the same assumptions as those imposed in prior work, in which case all inference procedures are asymptotically valid, and therefore amenable to juxtaposition. While beyond the scope of this paper, allowing for departures from asymptotic linearity is an interesting topic for future research.

\subsection{Standardized Statistics}\label{subsec: Standardized Statistics}

Suppressing the dependence on $\Lincom$ and $\Scale_\Lincom$, let
\begin{equation*}
    F(x) := \P\left[\frac{\widehat{\theta}_\Lincom - \theta_\Lincom}{\Scale_\Lincom} \leq x\right], \qquad x\in\R,
\end{equation*}
be the cumulative distribution function (cdf) of $(\widehat{\theta}_\Lincom - \theta_\Lincom)/\Scale_\Lincom$, where $\Scale_\Lincom$ is positive and non-random. Letting $\Phi$ denote the standard normal cdf, it follows from \eqref{eq:SB-Distribution} that
\begin{equation}\label{eq:SB-Distribution-restatement}
    \sup_{x\in \R}\left|F(x)- \Phi(x)\right| = o(1)
\end{equation}
under assumptions implying in particular that $\omega_\Lincom^2/\Scale_\Lincom^2 \to 1$, where
\begin{equation*}
    \omega_\Lincom^2 := \V[\widehat{\theta}_\Lincom] = n^{-1} \left[\sigma_\Lincom^2 + O\left(h^{P \land S}\right) \right] + \binom{n}{2}^{-1} h^{-d-2}[\delta_\Lincom^2 +O(h^2)],
\end{equation*}
with $\sigma_\Lincom^2 := \Lincom'\Sigma\Lincom$ and $\delta_\Lincom^2 := \Lincom'\Delta\Lincom$.

Our first theorem provides a refinement of \eqref{eq:SB-Distribution-restatement}. To state the theorem, let
\begin{equation*}
    \dot{f}_\Lincom(X) := \Lincom'\dot{f}(X), \qquad \varphi_\Lincom(Z) := \psi_\Lincom(Z) + 2\theta_\Lincom, \qquad \eta_\Lincom(Z_2) := \lim_{n\to\infty} \E[\varphi_\Lincom(Z_1)\Lincom'U_{12}|Z_2],
\end{equation*}
and define the following quantities (all of which are finite under the assumptions of Theorem \ref{thm:EE-standard}):
\begin{align*}
    \beta_\Lincom &:= 2(-1)^P \sum_{a \in \Z_+^d, [a] = P}\frac{\mu_a}{a!}\E\left[g(X) \frac{\partial^{a}}{\partial x^{a}}\dot{f}_\Lincom(X)\right], \qquad \kappa_{1,\Lincom} := \E[\psi_\Lincom(Z)^3],\\
    \kappa_{2,\Lincom} &:= \E[\varphi_\Lincom(Z)\dot{\eta_\Lincom}(Z)] - \E[\varphi_\Lincom(Z)^2]\theta_\Lincom - \V[\varphi_\Lincom(Z)]\theta_\Lincom.
\end{align*}
Also, let $\phi$ denote the standard normal probability density function.

\begin{thm}[Standardization] \label{thm:EE-standard}
    Suppose Assumptions \ref{A:DGP} and \ref{A:kernel} hold with $S \geq P$, and that $nh^{2P}\to 0$ and $nh^{d+2}/\log^9 n \to\infty$. If $\Scale_\Lincom$ is positive and non-random with $\omega_\Lincom^2/\Scale_\Lincom^2 \to 1$, then
    \begin{equation*}
        \sup_{x\in \R}\left|F(x)- G(x)\right| = o(r_n), \qquad r_n := \sqrt{n}h^P + \frac{1}{nh^{d+2}} + \frac{1}{\sqrt{n}},
    \end{equation*}
    with
    \begin{equation*}
        G(x) := \Phi(x) -\phi(x)\left[\frac{\sqrt{n}h^P\beta_\Lincom}{\sigma_\Lincom} + \frac{\omega_\Lincom^2/\Scale_\Lincom^2-1}{2} x + \frac{\kappa_{1,\Lincom}+\kappa_{2,\Lincom}}{6\sqrt{n}\sigma_\Lincom^3}(x^2-1)\right].
    \end{equation*}
\end{thm}

The proof of the theorem proceeds by verifying the high-level conditions of more general results presented in Appendix \ref{App A: Second-Order U-Statistics}. The general results establish a valid Edgeworth expansion for a generic class of U-statistics with $n$-varying kernels and may be of independent theoretical interest. Theorem \ref{thm:EE-standard} generalizes \citet[Theorem 1]{Nishiyama-Robinson_2000_ECMA} by allowing for a generic standardization factors $\Scale_\Lincom$ instead of their specific choice $\sqrt{\Lincom'\Sigma\Lincom/n} = \sigma_\Lincom/\sqrt{n}$. The latter generalization is important for our purposes, as it enables us to compare the different distributional approximations implied by \eqref{eq:AL-Distribution} and \eqref{eq:SB-Distribution}.

As is customary with Edgeworth expansions, the square-bracketed term in the function $G$ is a ``correction'' term capturing the extent to which the first three cumulants of the statistic differ from those of the standard normal distribution. To be specific, the first and third terms correct for bias and skewness, respectively. None of these correction terms depend on the particular $\Scale_\Lincom$ used for standardization purposes. In contrast, and as was to be expected, the variance correction term does depend on $\Scale_\Lincom$, being proportional to $\omega_\Lincom^2/\Scale_\Lincom^2 - 1$.

The asymptotic linearity result \eqref{eq:AL-Distribution} suggests setting $\Scale_\Lincom^2=\sigma_\Lincom^2/n$. Doing so, and in agreement with \citet[Theorem 1]{Nishiyama-Robinson_2000_ECMA}, we have
\begin{equation*}
    \omega_\Lincom^2/\Scale_\Lincom^2 - 1 \approx \frac{2\delta_\Lincom^2}{nh^{d+2}\sigma_\Lincom^2},
\end{equation*}
the approximation error being $o(r_n)$. In the display, the term on the right hand side involves $\delta_\Lincom^2 = \lim_{n\to\infty}h^{d+2}\V[\Lincom'Q_{ij}\Lincom]$ and is therefore interpretable as a variance correction term intrinsically associated with approximations based on asymptotic linearity, as such approximations ignore the contribution of the ``quadratic'' terms $Q_{ij}$ to the variability of $\widehat{\theta}$. Unlike \eqref{eq:AL-Distribution}, the small bandwidth formulation \eqref{eq:SB-Distribution} explicitly accounts for the presence of ``quadratic'' terms and the standardization factor $\Scale_\Lincom^2 = \V[\widehat{\theta}_\Lincom] = \omega_\Lincom^2$ suggested by the small bandwidth formulation is one for which the variance correction term in $G$ vanishes altogether.

Although the details of the results reported here are specific to DWAD estimation, one important qualitative conclusion appears to generalize: the feature that it can be advantageous (in a higher-order sense) to capture the full variability of a statistic when approximating its distribution is known to be shared by certain statistics arising in the context of nonparametric kernel-based density and local polynomial regression inference; for details, see \citet{Calonico-Cattaneo-Farrell_2018_JASA,Calonico-Cattaneo-Farrell_2022_Bernoulli}.

In isolation, Theorem \ref{thm:EE-standard} is mostly of theoretical interest, the reason being that it is concerned with standardized (as opposed to studentized) estimators. The consequences of employing studentization (i.e., replacing $\Scale_\Lincom$ with an estimator) will be explored in the next subsection. One important qualitative conclusion of that subsection concerns inference. That conclusion can be anticipated with the help of Theorem \ref{thm:EE-standard}. We conclude this subsection by doing so.

For any $\alpha\in(0,1)$, a natural (albeit infeasible) $100(1-\alpha)\%$ two-sided confidence interval for $\theta_\Lincom$ has endpoints given by $\widehat{\theta}_\Lincom \pm c_\alpha \Scale_\Lincom$, where $c_\alpha := \Phi^{-1}(1-\alpha/2)$ and where $\Scale_\Lincom^2$ is an approximate variance of $\widehat{\theta}_\Lincom$. Under the assumptions of Theorem \ref{thm:EE-standard}, the coverage probability of this interval satisfies
\begin{equation*}
    \P\left[\widehat{\theta}_\Lincom - c_\alpha \Scale_\Lincom \leq \theta_\Lincom \leq \widehat{\theta}_\Lincom + c_\alpha \Scale_\Lincom \right] = 1-\alpha - (\omega_\Lincom^2/\Scale_\Lincom^2 - 1) \phi(c_\alpha) c_\alpha + o(r_n),
\end{equation*}
so to the order considered the coverage error is proportional to the term $\omega_\Lincom^2/\Scale_\Lincom^2 - 1$ discussed previously and our conclusions about this term therefore apply directly. In particular, the coverage error of an infeasible interval using $\Scale_\Lincom^2 = \V[\widehat{\theta}_\Lincom]$ (as suggested by the small bandwidth asymptotic result \eqref{eq:SB-Distribution}) is $o(r_n)$, while the coverage errors of an infeasible intervals using $\Scale_\Lincom^2 = \sigma_\Lincom^2/n$ (as suggested by the asymptotic linearity result \eqref{eq:AL-Distribution}) or its pre-asymptotic counterpart $\Scale_\Lincom^2 = \V[\Lincom'L_i]/n$ are of larger magnitude.

\subsection{Studentized Statistics}\label{subsec: Studentized Statistics}

Next, we investigate the role of variance estimation by obtaining Edgeworth expansions for studentized versions of $\widehat{\theta}$. For specificity, and inspired by \eqref{eq:T-stat-AL} and \eqref{eq:T-stat-SB}, we compare
\begin{equation*}
    \widehat{F}_\AL(x) := \P\left[\frac{\widehat{\theta}_\Lincom - \theta_\Lincom}{\widehat{\Scale}_\mathtt{AL,\Lincom}} \leq x \right],\qquad \widehat{\Scale}_\mathtt{AL,\Lincom}^2 := \Lincom'\widehat{V}_\AL\Lincom
\end{equation*}
and
\begin{equation*}
    \widehat{F}_\SB(x) := \P\left[\frac{\widehat{\theta}_\Lincom - \theta_\Lincom}{\widehat{\Scale}_\mathtt{SB,\Lincom}} \leq  x \right],\qquad \widehat{\Scale}_\mathtt{SB,\Lincom}^2 := \Lincom'\widehat{V}_\SB\Lincom.
\end{equation*}
Studying $\widehat{F}_\AL$, \citet[Theorem 3]{Nishiyama-Robinson_2000_ECMA} found that if the assumptions of Theorem \ref{thm:EE-student} below are satisfied, then
\begin{equation*}
    \sup_{x\in \R}\left|\widehat{F}_\AL(x)- \widehat{G}_\AL(x)\right| = o(r_n),
\end{equation*}
with
\begin{equation*}
    \widehat{G}_\AL(x) := \Phi(x) -\phi(x)\left[\left( \frac{\sqrt{n}h^P\beta_\Lincom}{\sigma_\Lincom} - \frac{3\kappa_{1,\Lincom} + 2\kappa_{2,\Lincom}}{6\sqrt{n}\sigma_\Lincom^3}\right) - \frac{\delta_\Lincom^2}{nh^{d+2}\sigma_\Lincom^2} x - \frac{2\kappa_{1,\Lincom} + \kappa_{2,\Lincom}}{6\sqrt{n}\sigma_\Lincom^3} (x^2 - 1)\right],
\end{equation*}
where, once again, the three terms in square brackets correct for bias, variance, and skewness, respectively. In light of the results of the previous subsection, one would expect the Edgeworth approximation to $\widehat{F}_\SB$ to be similar to $\widehat{G}_\AL$, the only (possible) difference being the variance correction term. The following result shows that this is indeed the case.

\begin{thm}[Studentization]\label{thm:EE-student}
    Suppose Assumptions \ref{A:DGP} and \ref{A:kernel} hold with $S \geq P$ and $\E[Y^6] < \infty$, and that $nh^{2P}\to 0$ and $nh^{d+2}/\log^9 n \to\infty$. Then
    \begin{equation*}
        \sup_{x\in \R}\left|\widehat{F}_\SB(x)- \widehat{G}_\SB(x)\right| = o(r_n),
    \end{equation*}
    with
    \begin{equation*}
        \widehat{G}_\SB(x) := \Phi(x) -\phi(x)\left[\left(\frac{\sqrt{n}h^P\beta_\Lincom}{\sigma_\Lincom} - \frac{3\kappa_{1,\Lincom} + 2\kappa_{2,\Lincom}}{6\sqrt{n}\sigma_\Lincom^3}\right) - \frac{2\kappa_{1,\Lincom} + \kappa_{2,\Lincom}}{6\sqrt{n}\sigma_\Lincom^3} (x^2 - 1)\right]
    \end{equation*}

\end{thm}

This theorem shows that employing studentization based on small bandwidth asymptotics offers demonstrable improvements in terms of distributional approximations for the resulting feasible \textit{t}-test: the variance correction term present in $\widehat{G}_\AL$ is absent from $\widehat{G}_\SB$. As in \citet{Nishiyama-Robinson_2000_ECMA,Nishiyama-Robinson_2001_ChBook}, the result is obtained under the somewhat stronger moment condition $\E[Y^6] < \infty$ than the condition $\E[|Y|^3] < \infty$ of Theorem \ref{thm:EE-standard}, the purpose of the strengthened condition being to help control the contribution of the random denominator of the studentized version of $\widehat{\theta}$.

The main practical implication of Theorem \ref{thm:EE-student} can be illustrated by analyzing the coverage error of $100(1-\alpha)\%$ confidence intervals with endpoints $\widehat{\theta}_\Lincom \pm c_\alpha \widehat{\Scale}_\Lincom$. Setting $\widehat{\Scale}_\Lincom = \widehat{\Scale}_{\AL,\Lincom}$ and applying \citet[Theorem 3]{Nishiyama-Robinson_2000_ECMA}, we have
\begin{equation*}
    \P\left[\widehat{\theta}_\Lincom - c_\alpha \widehat{\Scale}_{\AL,\Lincom} \leq \theta_\Lincom \leq \widehat{\theta}_\Lincom + c_\alpha \widehat{\Scale}_{\AL,\Lincom} \right] = 1-\alpha + \frac{2\delta_\Lincom^2}{nh^{d+2}\sigma_\Lincom^2} \phi(c_\alpha) c_\alpha + o(r_n),
\end{equation*}
whereas setting $\widehat{\Scale}_\Lincom = \widehat{\Scale}_{\SB,\Lincom}$ and applying Theorem \ref{thm:EE-student} gives
\begin{equation*}
    \P\left[\widehat{\theta}_\Lincom - c_\alpha \widehat{\Scale}_{\SB,\Lincom} \leq \theta_\Lincom \leq \widehat{\theta}_\Lincom + c_\alpha \widehat{\Scale}_{\SB,\Lincom} \right] = 1-\alpha + o(r_n).
\end{equation*}
In other words, confidence intervals based on \eqref{eq:T-stat-SB} are demonstrably superior to those based on \eqref{eq:T-stat-AL} from a higher-order asymptotic point of view. This finding provides a theoretical explanation of the simulation evidence reported in \citet{Cattaneo-Crump-Jansson_2014a_ET,Cattaneo-Crump-Jansson_2014b_ET,Cattaneo-Crump-Jansson_2010_JASA}, where feasible confidence intervals based on small bandwidth asymptotics were shown to offer better finite sample performance in terms of coverage error than their counterparts based on classical asymptotic linear approximations.

\subsection{Discussion and Bootstrap-Based Inference}

In the previous subsections, we investigated the higher-order performance of large sample distributional approximations under two alternative asymptotic frameworks: asymptotic linearity and small bandwidth asymptotics. We found that the choice of studentization matters in terms of distributional approximation errors, even under conditions guaranteeing that asymptotic linearity holds ($nh^{d+2}\to\infty$), in which case both asymptotic frameworks are first-order valid. As a consequence, our Edgeworth expansions (reported in Theorems \ref{thm:EE-standard} and \ref{thm:EE-student}) provide alternative validation of the main conclusions obtained by \citet{Cattaneo-Crump-Jansson_2010_JASA,Cattaneo-Crump-Jansson_2014a_ET} using first-order distributional approximations: in terms of distributional approximation accuracy, the small bandwidth framework justifying \eqref{eq:SB-Distribution} and \eqref{eq:T-stat-SB} dominates the asymptotic linear framework justifying \eqref{eq:AL-Distribution} and \eqref{eq:T-stat-AL}.

It is natural to ask whether a similar ranking emerges when employing bootstrap-based inference procedures. \cite{Cattaneo-Crump-Jansson_2014b_ET} studied the first-order properties of the nonparametric bootstrap under small bandwidth asymptotics for the kernel-based DWAD estimator, and showed that a similar first-order pattern emerges in that case: Under slightly stronger assumptions than Assumptions \ref{A:DGP} and \ref{A:kernel}, they showed that if $(nh^{d+2} \land 1) nh^{2(P \land S) }\to 0$ and if $n^2 h^d \to \infty$, then
\begin{equation*}
    \V^*[\widehat{\theta}^*]^{-1/2}(\widehat{\theta}^*-\widehat{\theta})\rightsquigarrow_\P\m{N}(0,I),
\end{equation*}
where $\rightsquigarrow_\P$ denotes weak convergence in probability,
\begin{equation*}
    \V^*[\widehat{\theta}^*] = \V^*[\bar{L}^*] + \V^*[\bar{Q}^*], \qquad \V^*[\bar{L}^*] = n^{-1} \left[ \Sigma + o_\P(1) \right], \qquad \V^*[\bar{Q}^*] = 3\binom{n}{2}^{-1} h^{-d-2} \left[ \Delta + o_\P(1) \right],
\end{equation*}
and
\begin{equation*}
    n^{-1} \widehat{\Sigma}^* = n^{-1} [\Sigma + o_\P(1) ] + 4\binom{n}{2}^{-1} h^{-d-2}[\Delta+o_\P(1)], \qquad \widehat{\Delta}^* = \Delta + o_\P(1),
\end{equation*}
with $\widehat{\theta}^*$, $\bar{L}^*$, $\bar{Q}^*$, $\widehat{\Sigma}^*$ and $\widehat{\Delta}^*$ denoting nonparametric bootstrap analogs of $\widehat{\theta}$, $\bar{L}$, $\bar{Q}$, $\widehat{\Sigma}$ and $\widehat{\Delta}$, respectively, and $\V^*[\cdot]$ denoting the variance computed conditional on the original data. It follows from those results that under small bandwidth asymptotics the bootstrap consistently estimates the distribution of a studentized version of $\widehat{\theta}$ when $\widehat{V}_\SB$ is used for studentization purposes, but not when $\widehat{V}_\AL$ is. These conclusions are in perfect agreement with those discussed in Section \ref{sec: First-Order Distribution Theory}, the only notable difference in the details being that the variability induced by the bootstrap is larger outside the asymptotic linear regime because $\V^*[\bar{Q}^*]/\V[\bar{Q}]=3+o_\P(1)$. Furthermore, the second term of the bootstrap-based jackknife variance estimator $\widehat{\Sigma}^*$ is asymptotically doubled relative to the second term of the jackknife variance estimator $\widehat{\Sigma}$.

\citet{Nishiyama-Robinson_2005_ECMA} used Edgeworth expansions to study the properties of inference procedures based on the nonparametric bootstrap under assumptions implying asymptotic linearity. Based on the findings in this paper and those in \cite{Cattaneo-Crump-Jansson_2014b_ET}, we conjecture that analogous conclusions to those obtained herein will be valid for the case of bootstrap-based inference. While the conceptual parallelism between bootstrap-based inference and the results reported in this paper are clear, formalizing our conjecture requires substantial additional technical work due to the added complications associated with the data resampling, and hence we leave the theoretical analysis for future work.

\section{Conclusion}\label{sec: Conclusion}

Employing Edgeworth expansions, we compared the higher-order properties of two first-order distributional approximations and their associated confidence intervals for the kernel-based DWAD estimator of \citet{Powell-Stock-Stoker_1989_ECMA}. We showed that small bandwidth asymptotics not only give demonstrably better distributional approximations than those implied by asymptotic linearity, but also justifies employing a variance estimator for studentization purposes that improves the distributional approximation. The main takeaway from our results is that in two-step semiparametric settings, and related problems, alternative asymptotic approximations that capture higher-order terms, which are ignored by more traditional asymptotic linearity-based approximations, can deliver better distributional approximations and, by implication, more accurate inference procedures in finite samples. See \citet{Cattaneo-Jansson-Newey_2018_ET} for related discussion.

While beyond the scope of this paper, it would be of interest to develop analogous Edgeworth expansions for more general linear and non-linear two-step semiparametric procedures employing either Gaussian or resampling approximations under both conventional and alternative asymptotic frameworks \citep{Cattaneo-Crump-Jansson_2013_JASA,Cattaneo-Jansson_2018_ECMA,Cattaneo-Jansson-Ma_2019_RESTUD,Cattaneo-Jansson_2022_ET}. In particular, for the special case of kernel-based DWAD estimators, which is a linear two-step kernel-based semiparametric estimator, \citet{Nishiyama-Robinson_2005_ECMA} already obtained Edgeworth expansions for bootstrap-based inference procedures under asymptotic linearity that could be contrasted with those obtained under small bandwidth asymptotics \citep{Cattaneo-Crump-Jansson_2014b_ET}, after establishing more general Edgeworth expansions accounting for the bootstrap.

\appendix

\numberwithin{equation}{section}
\numberwithin{assumption}{section}
\numberwithin{lem}{section}
\numberwithin{thm}{section}
\numberwithin{coro}{section}

\section{Second-Order U-Statistics}\label{App A: Second-Order U-Statistics}

Let $U_n$ be a second-order U-statistic with $n$-varying kernel:
\begin{equation*}
    U_n :=\binom{n}{2}^{-1}\sum_{1\leq i< j\leq n} u_n(Z_i,Z_j),
\end{equation*}
where $Z_1,Z_2,\dots$ are $i.i.d.$ copies of a random vector $Z$ and where, for each $n \geq 2$, $u_n$ is a permutation symmetric $\R$-valued function. The main result of this section concerns the distribution of the (approximately) standardized statistic $(U_n-\Center_n)/\Scale_n$, where $\Center_n \in \R$ and $\Scale_n > 0$ are non-random. To be specific, dropping the subscript $n$ to simplify notation and defining
\begin{equation*}
    F(x):=\P\left[\frac{U-\Center}{\Scale}\leq x\right],\qquad x\in\R,
\end{equation*}
our objective is to obtain a valid Edgeworth expansion for $F$.

When stating and proving the result, it is useful to employ the Hoeffding decomposition
\begin{equation*}
    U-\Center = B + L + Q,
\end{equation*}
where
\begin{equation*}
    B := \E[U]-\Center, \qquad L := n^{-1} \sum_{i=1}^n\ell_i, \qquad Q := \binom{n}{2}^{-1}\sum_{1\leq i< j\leq n} q_{ij},
\end{equation*}
and where, for $1 \leq i < j \leq n$,
\begin{equation*}
    \ell_i := 2(\E[u(Z_i,Z_j)|Z_i] - \E[u(Z_i,Z_j)])
\end{equation*}
and
\begin{equation*}
    q_{ij} := u(Z_i,Z_i) - \frac{1}{2}(\ell(Z_i) + \ell(Z_j)) - \E[u(Z_i,Z_j)].
\end{equation*}
Also, it is useful to define $\sigma_{\ell}^2:=\E[\ell^2_1]$, $\sigma_{q}^2:=\E[q^2_{12}]$, $\varkappa_1:= \E[\ell_1^3]$, and $\varkappa_2:=\E[\ell_1\ell_2q_{12}]$, and to note that
\begin{equation*}
    \omega^2 := \V[U] = \V[L] + \V[Q] = n^{-1} \sigma_{\ell}^2 + \binom{n}{2}^{-1}\sigma_{q}^2 = n^{-1} \sigma_{\ell}^2 \left[1 + O(n^{-1} \sigma_q^2/\sigma_{\ell}^2)\right].
\end{equation*}

\begin{thm}\label{thm:EE-ustat}
    For some $p \in (2,3]$, let the following conditions hold:
    \begin{enumerate}[(a)]
        \item $\E\left[|\ell_1/\sigma_\ell|^3\right] =  O(1)$ and $\E[|q_{12}/\sigma_\ell|^p] < \infty$;
        \item $n^{-1}\sigma_\ell^2/\omega^2 \to 1$ and $\omega^2/\Scale^2 \to 1$;
        \item For every $\underline{c},\bar{c} > 0$,
        \begin{equation*}
            \limsup_{n \to \infty}\sup_{\underline{c} < |t| \leq \bar{c} \log n}|\E[\exp(\iota t\ell_1/\sigma_\ell)]| < 1.
        \end{equation*}
    \end{enumerate}
    Then
    \begin{equation*}
        \sup_{x\in\R}\left|F(x) - G(x)\right| = O \left( \m{E} + \frac{1}{\sqrt{n} \log n} \right),
    \end{equation*}
    where $G$ is the distribution function with characteristic function
    \begin{equation*}
        \chi_{G}(t):= \exp\left(\iota t\gamma_1-\frac{t^2}{2}\right)\left[ 1 + \sum_{2 \leq j \leq 9} \left(\iota t\right)^j\gamma_j \right],
    \end{equation*}
    with
    \begin{align*}
        \gamma_1 &:= \frac{B}{\Scale},\quad \gamma_2 := \frac{\omega^2-\Scale^2}{2\Scale^2},\quad \gamma_3 := \frac{\varkappa_1+6\varkappa_2}{6n^2\Scale^3}, \\
        \gamma_4 &:= \frac{\omega^2-\Scale^2}{4\Scale^4}\binom{n}{2}^{-1} \sigma_q^2, \quad \gamma_5 := \frac{1}{12 n^2\Scale^5}\left[\binom{n}{2}^{-1}\varkappa_1\sigma_q^2 + 6\left(n^{-1}\sigma_\ell^2-\Scale^2\right)\varkappa_2\right], \\
        \gamma_6 &:= \frac{1}{6n^4\Scale^6}\left[\varkappa_1\varkappa_2 + 12\binom{n}{2}^{-2}\binom{n}{4} \varkappa_2^2\right],\quad \gamma_7 := 0,\\
        \gamma_8 &:= \frac{n^{-1}\sigma_\ell^2-\Scale^2}{4n^4\Scale^8}\binom{n}{2}^{-2}\binom{n}{4} \varkappa_2^2, \quad \gamma_9 := \frac{1}{12n^6\Scale^9}\binom{n}{2}^{-2}\binom{n}{4} \varkappa_1\varkappa_2^2,
    \end{align*}
    and where
    \begin{align*}
        \m{E} &:= \frac{1}{n\sigma_\ell^4}\E[|\ell_1^2\ell_2q_{12}|] + \frac{1}{n^{3/2}\sigma_\ell^5}\E[|\ell_1^2\ell_2^2q_{12}|] + \frac{1}{n^{3/2}\sigma_\ell^3}\E[|\ell_1 q_{12}^2|] + \frac{1}{n^{3/2}\sigma_\ell^5}\E[|\ell_1\ell_2\ell_3q_{13}q_{23}|] \\
        &\qquad + \frac{1}{n^{3/2}\sigma_\ell^7}|\varkappa_2|\E[|\ell_1^2\ell_2q_{12}|] + \frac{1}{n^2\sigma_\ell^8}|\varkappa_2|\E[|\ell_1^2\ell_2^2q_{12}|] + \frac{\log^p n}{n^p \sigma_\ell^p} \m{M}_p(\log n),
    \end{align*}
    with
    \begin{equation*}
        \m{M}_p(m) := \left( mn \E[q_{12}^2] \right)^{p/2} + m \E\left[(n \E [q_{12}^2|Z_1])^{p/2}\right] + m n  \E\left[|q_{12}|^p\right]
    \end{equation*}
\end{thm}

\begin{coro}\label{coro:EE-ustat}
    If the assumptions of Theorem \ref{thm:EE-ustat} hold and if $\gamma_1 \to 0$, then
    \begin{equation*}
        \sup_{x\in\R}\left|F(x) - \bar{G}(x)\right| = O \left( \gamma_1^2 + \m{E} + \frac{1}{\sqrt{n} \log n} \right),
    \end{equation*}
    where $\bar{G}$ is the distribution function with characteristic function
    \begin{equation*}
        \chi_{\bar{G}}(t):= \exp\left(-\frac{t^2}{2}\right)\left[ 1 + \sum_{1 \leq j \leq 9} \left(\iota t\right)^j\gamma_j \right], \qquad \gamma_1:=B.
    \end{equation*}
\end{coro}

\begin{remark}
    For every $k \in \Z_+$, 
    \begin{equation*}
        \frac{1}{2\pi}\int_\R \exp\left( -\iota tx - \frac{t^2}{2}\right) (\iota t)^{k} \d t = \phi(x)H_k(x), \qquad k \in \Z_+,
    \end{equation*}
    where $H_k(x)$ is the $k$-th order Hermite polynomial (i.e., $H_0(k)=1$, $H_1(x) = x$, $H_2(x) = x^2-1$, and so on). Therefore, the characteristic function $\chi_{\bar{G}}$ from Corollary \ref{coro:EE-ustat} can be inverted to obtain the following closed form expression for $\bar{G}$:
    \begin{equation*}
        \bar{G}(x) = \Phi(x) - \phi(x)\left[\sum_{1 \leq j \leq 9} \gamma_j H_{j-1}(x)\right].
    \end{equation*}
\end{remark}

\begin{remark}
    Condition (c) of Theorem \ref{thm:EE-ustat} is implied by the following condition:
    \begin{enumerate}
        \item[(c$'$)] For every $\underline{c} > 0$,
        \begin{equation*}
            \limsup_{n \to \infty}\sup_{|t| > \underline{c}}|\E[\exp(\iota t\ell_1/\sigma_\ell)]| < 1.
        \end{equation*}
    \end{enumerate}
    Moreover, by the proposition following \citet[Lemma 1.4]{Petrov_1995_Book}, condition (c$'$) is in turn implied by the following condition:
    \begin{enumerate}
        \item[(c$''$)] $\limsup_{n,|t| \to \infty} \E[\exp(\iota t\ell_1/\sigma_\ell)]| < 1$.
    \end{enumerate}
    When $\ell_1$ does not depend on $n$, conditions (c) and (c$'$) are equivalent and it is customary to replace these conditions by (c$''$), which itself reduces to the familiar Cram\'er condition on $\ell_1$, namely
    \begin{equation}\label{Cramer condition}
        \limsup_{|t| \to \infty} \E[\exp(\iota t\ell_1)]| < 1.
    \end{equation}
    In contrast, when $\ell_1$ does depend on $n$, condition (c) is potentially easier to verify than (c$''$). 
    An example where this potential is realized is given by the DWAD estimator studied in Theorems \ref{thm:EE-standard} and \ref{thm:EE-student}. In that example, it appears difficult to formulate simple conditions under which (c$''$) holds, but we are able to verify (c) with the help of the following observation: By \eqref{E:exp_expansion}, (c) holds whenever there exists a fixed function $\lambda$ satisfying
    \begin{equation}\label{Approximate Cramer condition}
        \E[(\ell_1-\lambda(Z_1))^2] = o(1/\log^2 n) \qquad \text{and} \qquad \limsup_{|t| \to \infty} \E[\exp(\iota t \lambda(Z_1))]| < 1,
    \end{equation}
    a condition which is itself equivalent to \eqref{Cramer condition} when $\ell_1$ does not depend on $n$.

    To summarize, the condition \eqref{Approximate Cramer condition} is equivalent to (c$''$) when $\ell_1$ does not depend on $n$ and more generally it provides a simple sufficient condition for (c) when $\ell_1$ is mean square convergent. 
\end{remark}
    
\begin{remark}
    Suppose $u_n$ does not depend on $n$ and that
    \begin{equation*}
        \E[|\ell_1|^3] < \infty, \qquad \E[|q_{12}|^3] < \infty, \qquad \text{and} \qquad \limsup\limits_{|t|\to\infty}\left|\E[\exp(\iota t\ell_1)]\right|<1.
    \end{equation*}
    If $\Center=\E[u(Z_1,Z_2)]$ and if $\Scale^2 = n^{-1}\sigma_\ell^2$, then the assumptions of Corollary \ref{coro:EE-ustat} are satisfied with $\gamma_1=0$ and $\m{E}=O(n^{-1})$. Also, $\gamma_2=O(n^{-1})$ and $\gamma_j=O(n^{-1})$ for $4 \leq j \leq 9$, so
    \begin{equation*}
        \bar{G}(x) = \Phi(x) - \phi(x)\frac{\gamma_3}{6\sqrt{n}}(x^2-1) + O(n^{-1}), \qquad \gamma_3 = \frac{\varkappa_1+6\varkappa_2}{\sigma_\ell^3},
    \end{equation*}
    uniformly in $x \in \R$. In other words, we recover a variant of \citet[Theorem 1.2]{Bickel-Gotze-vanZwet_1986_AOS}. See also \citet{Jing-Wang_2003_AOS}.
\end{remark}

\begin{remark}
    If $\E\left[|\ell_1/\sigma_\ell|^3\right] =  O(1)$ and if $n^{-1}\sigma_q^2/\sigma_\ell^2 \to 0$, then the H\"older inequality implies
    \begin{align*}
        \m{E}
        &\ls \frac{1}{n}\sqrt[3]{\E[|q_{12}/\sigma_\ell|^3]} + \left(\frac{1}{n^{3/4}}\sqrt[3]{\E[|q_{12}/\sigma_\ell|^3]}\right)^2 \\
        &+ \left(\frac{\log^3 n}{n} \E\left[|q_{12}/\sigma_\ell|^2\right] \right)^{p/2} + \frac{\log^{1+p} n}{n^{p/2}}  \E\left[(\E [(q_{12}/\sigma_\ell)^2|Z_1])^{p/2}\right] + \frac{\log^{1+p} n}{n^{p-1}} \E\left[|q_{12}/\sigma_\ell|^p\right],
    \end{align*}
    where $a\ls b$ denotes $a\leq C b$ for some positive constant $C$.
    Thus, for $p = 3$ the majorant side of $\m{E}$ is $o(1)$ iff
    \begin{equation*}
        \E\left[|q_{12}/\sigma_\ell|^2\right] = o(n/\log^3 n), \qquad     \E\left[|q_{12}/\sigma_\ell|^3\right] = o(n^2/\log^4 n),
    \end{equation*}
    and if
    \begin{equation*}
        \E\left[(\E [(q_{12}/\sigma_\ell)^2|Z_1])^{3/2}\right] = o(n^{3/2}/\log^4 n).
    \end{equation*}
\end{remark}

\section{Proof of Theorem \ref{thm:EE-ustat}}

Before presenting the formal proof, we outline the main steps involved, and compare our proof strategy to the approach taken in \citet{Jing-Wang_2003_AOS} for second-order U-statistics with fixed kernels (i.e., $u_n$ not depending on $n$) and, more broadly, to the classical Edgeworth expansion theory for sums of independent random variables \citep[e.g.,][]{Bhattacharya-Rao_1976_Book,Hall_1992_Book}.

We start from the following bound on the Kolmogorov distance between $F$ and $G$:
\begin{equation*}
    \rho(F,G) := \sup_{t\in\R}\left|F(t) - G(t)\right| \ls \int_{|t| \leq \sqrt{n}\log n } \left |\frac{\chi_F(t)-\chi_G(t)}{t}\right|\d t + \frac{1}{\sqrt{n}\log n},
\end{equation*}
where $\chi_F$ is the characteristic function of $F$. The integral is then upper bounded over three different frequency domains: Low frequency (LF), $|t| \leq \log n$; medium frequency (MF), $\log n < |t| \leq c \sqrt{n}$ (for a judiciously chosen $c$); and high frequency (HF), $c \sqrt{n} < |t| \leq \sqrt{n}\log n$. In the case of MF and HF, we further use the bound $|\chi_F(t)-\chi_G(t)|\leq|\chi_F(t)|+|\chi_G(t)|$ (i.e., the triangle inequality) and deal with each term separately. At this level of generality the proof strategy is similar to the canonical case of first-order U-statistics \citep[e.g.,][]{Bhattacharya-Rao_1976_Book,Hall_1992_Book}. However, we need to proceed differently to control the influence of the quadratic term, for which a bound follows almost exclusively from the linear part. For instance, for HF, we use Assumption (c) to find a $b>0$ such that for large $n$ and for $|t| \in (c \sqrt{n}, \sqrt{n} \log n]$,
\begin{equation*}
    \left|\chi_\ell\left(\frac{t}{n \Scale}\right)\right|\leq 1-b<\exp(-b), \qquad \chi_\ell(t) := \E[\exp(\iota t\ell_1)].
\end{equation*}

For the most interesting part, the LF domain, we begin by approximating $\chi_F(t)$ as follows:
\begin{align*}
    \chi_F(t) &= \exp\left(\iota \frac{t}{\Scale} B\right) \E\left[\exp\left(\iota\frac{t}{\Scale}L\right)\exp\left(\iota \frac{t}{\Scale}Q\right)\right] \\
    &\approx \exp\left(\iota \frac{t}{\Scale} B\right) \E\left[\left(\iota\frac{t}{\Scale}L\right)\left(1 + \iota \frac{t}{\Scale}Q - \frac{t^2}{2\Scale^2}Q^2\right)\right],
\end{align*}
and then examine each term on the right hand side separately while controlling the approximation error using Lemma \ref{L:Q_Lp} and the bound
\begin{equation}\label{E:exp_expansion}
    \left|\exp(\iota x) - \sum_{0 \leq j \leq 2} \frac{(\iota x)^j}{j!} \right|\leq  |x|^p.
\end{equation}
When doing this, our proof departs from \citet{Jing-Wang_2003_AOS} because we do not know \emph{ex-ante} which term(s) are higher-order due to the possibly $n$-varying structure of the U-statistic kernel. Therefore, we keep track of all terms, with special attention to the contribution of the terms involving $Q$. Finally, we collect all (possibly) leading terms in the approximation to $\chi_F$ in $\chi_G$ and arrive at a result of the form
\begin{equation}\label{E:cf_approximation}
    |\chi_F(t)-\chi_G(t)| \ls  \exp\left(-\frac{t^2}{4}\right)\m{R}(t) +  \frac{|t|^p}{n^{3p/2} \Scale^p \log^{p/2} n} \m{M}_p(\log n),
\end{equation}
where
\begin{align*}
    \m{R}(t) := &\frac{t^4}{n^3 \Scale^4}\E[|\ell_1^2\ell_2q_{12}|] + \frac{|t|^5}{n^4 \Scale^5}\E[|\ell_1^2\ell_2^2q_{12}|] + \frac{|t|^3}{n^3 \Scale^3}\E[|\ell_1 q_{12}^2|] + \frac{|t|^5}{n^4 \Scale^5}\E [|\ell_1\ell_2\ell_3q_{13}q_{23}|] \\
    &+ \frac{|t|^7}{n^5 \Scale^7}|\varkappa_2|\E[|\ell_1^2\ell_2q_{12}|]+ \frac{t^8}{n^6 \Scale^8}|\varkappa_2|\E[|\ell_1^2\ell_2^2q_{12}|],
\end{align*}
and where $n^{-1}\sigma_\ell^2/\Scale^2 \to 1$.

\bigskip
\textbf{Technical Details}.
Letting $g$ denote the Lebesgue density of $G$, application of a ``smoothing inequality'' \citep[e.g., Theorem 5.1 of][]{Petrov_1995_Book} gives
\begin{equation*}
    \rho(F,G) \ls \int_{|t| \leq \upsilon} \left |\frac{\chi_F(t)-\chi_G(t)}{t}\right|\d t +\frac{\sup_{x\in\R}|g(x)|}{\upsilon},\qquad \upsilon>0.
\end{equation*}
Setting $v=\sqrt{n}\log n$ and using the fact that $g$ is bounded (because $\sum_{2 \leq j \leq 9} |\gamma_j| \to 0$), it follows from the triangle inequality that
\begin{equation}\label{E:inter_bond}
    \rho(F,G) \ls \m{I}_1 + \m{I}_2 + \m{I}_3 + \m{I}_4 +\frac{1}{\sqrt{n}\log n},
\end{equation}
where, with $c>0$ a constant to be chosen later,
\begin{align*}
    \m{I}_1&:=\int_{|t|\leq\log n}\left |\frac{\chi_F(t)-\chi_G(t)}{t}\right|\d t, \qquad \m{I}_2:=\int_{\log n<|t|\leq c \sqrt{n}}\left |\frac{\chi_F(t)}{t}\right|\d t, \\
    \m{I}_3&:=\int_{c \sqrt{n}<|t|\leq \sqrt{n}\log n}\left |\frac{\chi_F(t)}{t}\right|\d t, \qquad \m{I}_4:=\int_{|t|>\log n}\left |\frac{\chi_G(t)}{t}\right|\d t.
\end{align*}
In what follows, we bound each of these integrals in turn.

\bigskip
\textit{Bound for $\m{I}_1$}.
We start by approximating
\begin{equation*}
    \chi_{L+Q}\left(\frac{t}{\Scale}\right) := \E\left[\exp \left(\iota \frac{t}{\Scale}L\right) \exp \left(\iota \frac{t}{\Scale}Q\right)\right].
\end{equation*}
First, using \eqref{E:exp_expansion}, we have
\begin{equation*}
    \left|\chi_{L+Q}\left(\frac{t}{\Scale}\right) - \E\left[\exp \left(\iota \frac{t}{\Scale}L\right)\left(1 + \iota \frac{t}{\Scale}Q - \frac{t^2}{2\Scale^2}Q^2\right)\right]\right| \leq |\tau|^p n^p \E\left[|Q|^p\right], \qquad  \tau := \frac{t}{n \Scale}.
\end{equation*}
Also, since $\ell_1,\dots,\ell_n$ are $i.i.d.$,
\begin{equation*}
    \E\left[\exp  \left(\iota \frac{t}{\Scale}L\right)\right] = \E\left[\exp \left(\iota \tau \sum_{1 \leq i \leq n} \ell_i\right)\right] = \E\left[\prod_{1 \leq i \leq n} \exp \left(\iota \tau \ell_i\right)\right] = \chi_\ell(\tau)^n.
\end{equation*}
Finally, because $\ell_k$ is independent of $q_{ij}$ when $k\not\in\{i,j\}$,
\begin{align*}
    \binom{n}{2} \E\left[\exp \left(\iota \frac{t}{\Scale}L\right)Q\right] &= \sum_{1 \leq i < j \leq n} \E\left[ q_{ij} \prod_{1 \leq k \leq n} \exp \left( \iota \tau \ell_k \right)\right] \\
    &= \sum_{1 \leq i < j \leq n} \E\left[ \exp\left( \iota \tau (\ell_i+\ell_j) \right) q_{ij} \prod_{\substack{1 \leq k \leq n \\ k \not \in \{i,j\}}} \exp \left(\iota\tau\ell_k\right)\right] \\
    &= \binom{n}{2} \chi_\ell(\tau)^{n-2}\E\left[\exp\left(\iota\tau(\ell_1+\ell_2)\right)q_{12}\right]
\end{align*}
and
\begin{align*}
    &\binom{n}{2}^2 \E\left[\exp \left(\iota \frac{t}{\Scale}L\right)Q^2\right] \\
    &= \sum_{1 \leq i < j \leq n}\E\left[q_{ij}^2\prod_{1 \leq m \leq n} \exp \left(\iota\tau\ell_m\right)\right] + \sum_{1 \leq i < j < l \leq n}\E\left[q_{ij}q_{jl}\prod_{1 \leq m \leq n} \exp \left(\iota\tau\ell_m\right)\right] \\
    &\qquad + \sum_{1 \leq i < j < k < l \leq n}\E\left[q_{ij}q_{kl}\prod_{1 \leq m \leq n} \exp \left(\iota\tau\ell_m\right)\right] \\
    &= \binom{n}{2}\chi_\ell(\tau)^{n-2}\E\left[\exp\left(\iota\tau(\ell_1 +\ell_2)\right)q_{12}^2\right] +\binom{n}{3}\chi_\ell(\tau)^{n-3}\E\left[\exp\left(\iota\tau(\ell_1 +\ell_2+\ell_3)\right)q_{12}q_{23}\right] \\
    &\qquad +\binom{n}{4}\chi_\ell(\tau)^{n-4}\left(\E\left[\exp\left(\iota\tau(\ell_1 +\ell_2)\right)q_{12}\right]\right)^2.
\end{align*}
Using the four preceding displays, we therefore have, uniformly in $|t| \leq \log n$,
\begin{align}\label{E:second_expansion}
    \chi_{L+Q}\left(\frac{t}{\Scale}\right) &= \chi_\ell(\tau)^n \nonumber \\
    &\quad + \chi_\ell(\tau)^{n-2} \left[\iota \tau n\E\left[ \exp(\iota \tau(\ell_1+\ell_2))q_{12}\right] - \frac{\tau^2}{2} n^2\binom{n}{2}^{-1}\E\left[\exp(\iota \tau(\ell_1+\ell_2))q_{12}^ 2\right]\right] \nonumber \\
    &\quad - \chi_\ell(\tau)^{n-3}\frac{\tau^2}{2}n^2\binom{n}{2}^{-2}\binom{n}{3}\E\left[\exp(\iota \tau(\ell_1+\ell_2+\ell_3))q_{13}q_{23}\right] \nonumber\\
    &\quad - \chi_\ell(\tau)^{n-4}\frac{\tau^2}{2}n^2\binom{n}{2}^{-2}\binom{n}{4} \left(\E\left[ \exp(\iota \tau(\ell_1+\ell_2))q_{12}\right]\right)^2 \nonumber\\
    &\quad  + |\tau|^p O\left(n^p \E\left[|Q|^p\right] \right).
\end{align}

Next, using degeneracy of $q_{ij}$ and \eqref{E:exp_expansion}, we have
\begin{align*}
    \E\left[\exp(\iota \tau(\ell_1+\ell_2))q_{12}\right]
    &= - \tau^2 \E\left[\ell_1\ell_2q_{12}\right]  \\
    &\quad + \iota \tau \E\left[\ell_1\left(\exp(\iota \tau\ell_2)-1-\iota \tau\ell_2\right)q_{12} + \ell_2\left(\exp(\iota \tau\ell_1)-1-\iota \tau\ell_1\right)q_{12}\right] \\
    &\quad + \E\left[\left(\exp(\iota \tau\ell_1)-1-\iota \tau\ell_1\right)\left(\exp(\iota \tau\ell_2)-1-\iota \tau\ell_2\right)q_{12}\right] \\
    & = - \tau^2\varkappa_2 + |\tau|^3 O\left(\E\left[|\ell_1^2\ell_2q_{12}|\right]\right) + \tau^4 O\left(\E\left[|\ell_1^2\ell_2^2q_{12}|\right]\right).
\end{align*}
Similarly,
\begin{align*}
    \E\left[\exp(\iota \tau(\ell_1+\ell_2))q_{12}^ 2\right] &= \E\left[q_{12}^ 2\right] +  \E \left[\left(\exp(\iota \tau(\ell_1+\ell_2))-1\right)q_{12}^2\right] \\
    &= \sigma_q^2 + |\tau| O\left(\E\left[|\ell_1q_{12}^2|\right]\right)
\end{align*}
and
\begin{align*}
    \E\left[\exp(\iota \tau(\ell_1+\ell_2+\ell_3))q_{13}q_{23}\right] &= \E\left[\left( \prod_{1 \leq i \leq 3} \left(\exp(\iota \tau \ell_i)-1\right)\right)q_{13}q_{23}\right] \\
    &= |\tau|^3 O\left(\E\left[|\ell_1\ell_2\ell_3q_{13}q_{23}|\right]\right).
\end{align*}
Also, using arguments familiar from the Edgeworth expansion theory for sum of $i.i.d.$ random variables \citep[e.g.,][]{Bhattacharya-Rao_1976_Book,Hall_1992_Book}, we have, for $k \in \{0,2,3,4\}$,
\begin{align*}\label{E:iid_EE}
    \chi_\ell(\tau)^{n-k} &= \exp\left(-\frac{t^2}{2}\right)\left[1 - \left(\frac{n^{-1}\sigma_\ell^2 - \Scale^2}{\Scale^2}\right)\frac{t^2}{2} - \iota \frac{\varkappa_1}{n^2 \Scale^3} \frac{t^3}{6}+ \left(\frac{n^{-1}\sigma_\ell^2 - \Scale^2}{\Scale^2}\right)^2 O(t^4) \right]  \\
    &\qquad + \exp\left(-\frac{t^2}{4}\right) o\left( \frac{|t|^3}{\sqrt{n}} + \frac{t^6}{\sqrt{n}} \right).
\end{align*}
Finally, using Lemma \ref{L:Q_Lp}, we have
\begin{equation*}
    \E\left[|Q|^p\right] \ls \frac{\m{M}_p(n)}{n^{2p}}  \ls \frac{\m{M}_p(\log n)}{n^{3p/2} \log^{p/2} n}.
\end{equation*}

Plugging the displays from the previous paragraph into \eqref{E:second_expansion}, we obtain \eqref{E:cf_approximation} and   therefore
\begin{align*}
    \m{I}_1 &\ls \int_{|t|\leq\log n}\exp\left(-\frac{t^2}{4}\right)|t|^{-1}\m{R}(t)\d t + \frac{\m{M}_p(\log n)}{n^{3p/2} \Scale^p \log^{p/2} n } \int_{|t|\leq\log n} |t|^{p-1}\d t \\
    &\ls \m{R}(1) + \frac{\log^{p/2} n}{n^{3p/2} \Scale^p} \m{M}_p(\log n) \\
    &\ls \frac{1}{n\sigma_\ell^4}\E[|\ell_1^2\ell_2q_{12}|] + \frac{1}{n^{3/2}\sigma_\ell^5}\E[|\ell_1^2\ell_2^2q_{12}|] + \frac{1}{n^{3/2}\sigma_\ell^3}\E[|\ell_1 q_{12}^2|] + \frac{1}{n^{3/2}\sigma_\ell^5}\E[|\ell_1\ell_2\ell_3q_{13}q_{23}|] \\
    &\qquad + \frac{1}{n^{3/2}\sigma_\ell^7}|\varkappa_2|\E[|\ell_1^2\ell_2q_{12}|] + \frac{1}{n^2\sigma_\ell^8}|\varkappa_2|\E[|\ell_1^2\ell_2^2q_{12}|] + \frac{\log^{p/2} n}{n^p \sigma_\ell^p} \m{M}_p(\log n),
\end{align*}
where the last $\ls$ uses $n^{-1}\sigma_\ell^2/\Scale^2 \to 1$.

\bigskip
\textit{Bound for $\m{I}_2$}.
For $1\leq m <n$, defining
\begin{equation*}
    Q(m) := \binom{n}{2}^{-1} \sum_{\substack{1 \leq i < j \leq n \\ i \leq m}} q_{ij}
\end{equation*}
and using \eqref{E:exp_expansion}, we have
\begin{align*}
    |\chi_F(t)| = \left|\chi_{L+Q}\left(\frac{t}{\Scale}\right)\right| \leq \left| \E\left[\exp\left(\iota \frac{t}{\Scale}(L+Q-Q(m))\right)\sum_{0 \leq k \leq 2}\frac{(\iota t)^k}{k!\Scale^k}Q(m)^k\right]\right| + \frac{|t|^p}{\Scale^p} \E\left[|Q(m)|^p\right],
\end{align*}
where, using the fact that $Q - Q(m)$ is a function of $X_{m+1},\dots, X_n$, it can be shown that 
\begin{align*}
    \left| \E\left[\exp\left(\iota \frac{t}{\Scale}(L+Q-Q(m))\right)Q(m)^k\right]\right| \ls \left(\frac{m}{n}\right)^k |\chi_\ell(\tau)|^{m-2k} \E[|q_{12}|^k], \qquad k\in\{0,1,2\},
\end{align*}
and therefore, using $n^{-1}\sigma_\ell^2/\Scale^2 \to 1$,
\begin{equation}\label{E:large_deviation_bound}
    |\chi_F(t)| \ls \sum_{0 \leq k \leq 2} \left(\frac{|t|m}{\sqrt{n}}\right)^k |\chi_\ell(\tau)|^{m-2k} \E\left[\left|\frac{q_{12}}{\sigma_\ell}\right|^k\right] + \frac{|t|^p}{\Scale^p} \E[|Q(m)|^p].
\end{equation}

Next, because $\E\left[|\ell_1/\sigma_\ell|^3\right] =  O(1)$ and $n^{-1}\sigma_\ell^2/\Scale^2 \to 1$, there exists a $c>0$ such that, for $n$ sufficiently large,
\begin{equation*}
    |\chi_\ell(\tau)| \leq 1 - \frac{t^2}{3n} \leq \exp\left(-\frac{t^2}{3n}\right), \qquad t \leq c \sqrt{n}.
\end{equation*}
Setting $m = \lfloor 15 n\log n/t^2\rfloor$ in \eqref{E:large_deviation_bound}, where $\lfloor \cdot \rfloor$ denotes the floor operator, and using the preceding display, we obtain (for $n$ sufficiently large)
\begin{equation*}
    |\chi_F(t)| \ls \sum_{0 \leq k \leq 2} \frac{|t|^k}{n^{5-k}} \E\left[\left|\frac{q_{12}}{\sigma_\ell}\right|^k\right] + \frac{|t|^p}{\Scale^p} \E[|Q(m)|^p], \qquad \log n <|t|\leq c\sqrt{n},
\end{equation*}
where, using Lemma \ref{L:Q_Lp} and $n^{-1}\sigma_\ell^2/\Scale^2 \to 1$,
\begin{align*}
    \frac{|t|^p}{\Scale^p} \E[|Q(m)|^p] &\ls \frac{|t|^p}{n^{2p} \Scale^p } \m{M}_p(m) \\
    &\ls \frac{\log^{p/2} n}{n^{p/2} \sigma_\ell^p} \sigma_q^p + |t|^{p-2} \frac{\log n}{n^{p-1} \sigma_\ell^p} \left[ \E\left[(\E [q_{12}^2|Z_1])^{p/2}\right] + \frac{\E\left[|q_{12}|^p\right]}{n^{p/2-1}}  \right], \qquad \log n <|t|\leq c\sqrt{n}.
\end{align*}

As a consequence, using $n^{-1}\sigma_q^2/\sigma_\ell^2 \to 0$,
\begin{align*}
    \m{I}_2 &\ls \sum_{0 \leq k \leq 2} \frac{1}{n^{5-3k/2}} \E\left[\left|\frac{q_{12}}{\sigma_\ell}\right|^k\right] + \frac{\log^{1+p/2} n}{n^{p/2} \sigma_\ell^p} \sigma_q^p + \frac{\log n}{n^{p/2} \sigma_\ell^p} \E\left[(\E [q_{12}^2|Z_1])^{p/2}\right] + \frac{\log n}{n^{p-1} \sigma_\ell^p} \E\left[|q_{12}|^p\right] \\
    &= o(n^{-1}) +  o\left(\frac{\log^p n}{n^p \sigma_\ell^p} \m{M}_p(\log n)\right).
\end{align*} 

\bigskip
\textit{Bound for $\m{I}_3$}. By condition (c), there exists $b > 0$ such that, for $n$ sufficiently large,
\begin{equation*}
    |\chi_\ell(\tau)|\leq 1-b \leq \exp(-b), \qquad c \sqrt{n} < |t| < \sqrt{n} \log n.
\end{equation*}
Setting $m = \lfloor 4\log n/b\rfloor$ in \eqref{E:large_deviation_bound} and using the preceding display, we obtain (for $n$ sufficiently large)
\begin{equation*}
    |\chi_F(t)| \ls \sum_{0 \leq k \leq 2} \frac{|t|^k\log^k n}{n^{4+k/2}} \E\left[\left|\frac{q_{12}}{\sigma_\ell}\right|^k\right] + \frac{|t|^p}{\Scale^p} \E[|Q(m)|^p], \qquad \qquad c \sqrt{n} < |t| < \sqrt{n} \log n,
\end{equation*}
where, using Lemma \ref{L:Q_Lp} and $n^{-1}\sigma_\ell^2/\Scale^2 \to 1$,
\begin{align*}
    \frac{|t|^p}{\Scale^p} \E[|Q(m)|^p] &\ls \frac{|t|^p}{n^{2p} \Scale^p } \m{M}_p(m) \ls \frac{|t|^p}{n^{3p/2} \sigma_\ell^p} \m{M}_p(\log n), \qquad \qquad c \sqrt{n} < |t| < \sqrt{n} \log n.
\end{align*}

As a consequence, using $n^{-1}\sigma_q^2/\sigma_\ell^2 \to 0$,
\begin{align*}
    \m{I}_3 &= o(n^{-2}) + O\left(\frac{\m{M}_p(\log n)}{n^{3p/2} \sigma_\ell^p} \int_{c\sqrt{n}\leq |t|\leq \sqrt{n}\log n} |t|^{p-1} dt\right) \\
    &=  o(n^{-2}) + O\left(\frac{\log^p n}{n^p \sigma_\ell^p} \m{M}_p(\log n) \right).
\end{align*} 

\bigskip
\textit{Bound for $\m{I}_4$}.
For every $j$, we have
\begin{equation*}
    \int_{t>\log n}t^{j-1}\exp\left(-\frac{t^2}{2}\right) \d t \ls \int_{t>\log n}\exp\left(-\frac{t^2}{4}\right) \d t = o(n^{-1}),
\end{equation*}
and therefore, using $\sum_{2 \leq j \leq 9} |\gamma_j| \to 0$,
\begin{align*}
    \m{I}_4&\ls \int_{|t|>\log n}|t|^{-1}\exp\left(-\frac{t^2}{2}\right)\left|1 +\sum_{j=2}^9\left(\iota t\right)^j\gamma_j \right|\d t \\
    &\ls \left(1 + \sum_{j=2}^9 |\gamma_j|\right) \int_{t>\log n}t^{-1}\exp\left(-\frac{t^2}{4}\right)\d t = o(n^{-1}).
\end{align*}
\qed

\section{Proof of Theorem \ref{thm:EE-standard}}\label{App B: Proof of Theorem EE-standard}

We employ Corollary \ref{coro:EE-ustat} with $u(Z_i,Z_j) = \Lincom'U_{ij}$ and $p=3$. Proceeding as in \citet{Cattaneo-Crump-Jansson_2010_JASA,Cattaneo-Crump-Jansson_2014a_ET,Cattaneo-Crump-Jansson_2014b_ET}, condition (a) of Theorem \ref{thm:EE-ustat} can be verified by direct calculations. Also, condition (b) of Theorem \ref{thm:EE-ustat} holds because
\begin{equation*}
    \sigma_\ell^2 = \sigma_\Lincom^2 + o(1), \qquad \sigma_q^2 = \frac{\delta_\Lincom^2 + o(1)}{h^{d+2}}, \qquad \text{and} \qquad nh^{d+2} \to \infty,
\end{equation*}
while condition (c) of Theorem \ref{thm:EE-ustat} holds because \eqref{Approximate Cramer condition} is satisfied with $\lambda = \psi_\Lincom$. The additional condition $\gamma_1 \to 0$ of Corollary \ref{coro:EE-ustat} holds if $nh^{2P}\to 0$ because it follows from routine (bias) calculations that $\E[\widehat{\theta}_\Lincom] - \theta_\Lincom = h^P\beta_\Lincom + o(h^P)$.

Next, the law of iterated expectations, integration by parts, and Taylor series expansions can be used to show that
\begin{equation*}
    \varkappa_1 = \kappa_{1,\Lincom} + O(h^P), \qquad \text{and} \qquad \varkappa_2 = \kappa_{2,\Lincom} + O(h^P),
\end{equation*}
and also that $\gamma_4 \ls n^{-3}h^{-d-2}$, $\gamma_5 \ls n^{-2}$, $\gamma_6 \ls n^{-1}$, $\gamma_8 \ls n^{-3}$, and $\gamma_9 \ls n^{-7/2}$.

Finally, by \citet[Supplemental Appendix]{Cattaneo-Crump-Jansson_2014b_ET}, we have
\begin{equation*}
    \E [q_{12}^2]^{3/2} \leq \E\left[(\E [q_{12}^2|Z_1])^{3/2}\right] \ls h^{-3d/2-3} \qquad \text{and} \qquad \E[|q_{12}|^3] \ls h^{-2d-3}.
\end{equation*}
Using these bounds and the H\"older inequality, we find that $\m{E} = o(n^{-1}h^{-d-2})$.
\qed

\section{Proof of Theorem \ref{thm:EE-student}}\label{App C: Proof of Theorem EE-student}

Letting $\ell_i := \Lincom'L_i$ and $q_{ij} := \Lincom'Q_{ij}$, occasionally suppressing the dependence on $\Lincom$, and using the Hoeffding decomposition of $\widehat{\theta}_\Lincom -\theta_\Lincom$ along with the identity
\begin{equation}\label{eq:linearization-Student}
    \frac{\Scale}{\widehat{\Scale}} = 1 - \frac{\widehat{\Scale}^2 - \Scale^2}{2\Scale^2} + \frac{(\widehat{\Scale}+2\Scale)(\widehat{\Scale}^2 - \Scale^2)^2}{2\Scale^2 \widehat{\Scale}(\widehat{\Scale} + \Scale)^2},
\end{equation}
we have
\begin{equation*}
    \frac{\widehat{\theta}_\Lincom -\theta_\Lincom}{\widehat{\Scale}_{\SB,\Lincom}} = \widetilde{T}_\SB + R_\SB,
\end{equation*}
where
\begin{equation*}
    \widetilde{T}_\SB := \frac{B + L + Q}{\Scale_{\SB,\Lincom}} - \frac{\widehat{\Scale}_{\SB,\Lincom}^2 - \Scale_{\SB,\Lincom}^2}{2\Scale_{\SB,\Lincom}^2} \frac{L + Q}{\Scale_{\SB,\Lincom}},
\end{equation*}
with $\Scale_{\SB,\Lincom}$ is a judiciously chosen positive scalar,
\begin{equation*}
    B := \E[\widehat{\theta}_\Lincom]-\theta_\Lincom, \qquad L := n^{-1} \sum_{1 \leq i \leq n} \ell_i, \qquad \text{and} \qquad Q := \binom{n}{2}^{-1}\sum_{1 \leq i < j \leq n} q_{ij},
\end{equation*}
and where
\begin{equation*}
    R_\SB := -  \frac{\widehat{\Scale}_{\SB,\Lincom}^2 - \Scale_{\SB,\Lincom}^2}{2\Scale_{\SB,\Lincom}^2} \frac{B}{\Scale_{\SB,\Lincom}} +  \frac{(\widehat{\Scale}_{\SB,\Lincom}+2\Scale_{\SB,\Lincom})(\widehat{\Scale}_{\SB,\Lincom}^2 - \Scale_{\SB,\Lincom}^2)^2}{2\Scale_{\SB,\Lincom}^2 \widehat{\Scale}_{\SB,\Lincom}(\widehat{\Scale}_{\SB,\Lincom} + \Scale_{\SB,\Lincom})^2} \frac{B + L + Q}{\Scale_{\SB,\Lincom}}
\end{equation*}
is a remainder term.

Defining
\begin{equation*}
    \widetilde{F}_\SB(x) := \P\big[\widetilde{T}_\SB \leq x \big]
\end{equation*}
and adapting the proof of Theorem \ref{thm:EE-ustat}, we obtain
\begin{equation}\label{eq:EE-Student}
    \rho(\widetilde{F}_\SB,\widehat{G}_\SB) = o(r_n)
\end{equation}
by applying a smoothing inequality followed by a split of the frequency domain of the resulting integral, where bounding the various integrals requires some additional care due to the presence of the variance estimator.

Also, using the strengthened moment condition $\E[Y^6] < \infty$, we obtain
\begin{equation}\label{eq:Remainder-Student}
    \P[|R_\SB| > r_n/\log n] = o(r_n),
\end{equation}
implying in turn that $\rho(\widehat{F}_\SB,\widetilde{F}_\SB) = o(r_n)$ and therefore also that
\begin{equation*}
    \rho(\widehat{F}_\SB,\widehat{G}_\SB) \leq \rho(\widehat{F}_\SB,\widetilde{F}_\SB) + \rho(\widetilde{F}_\SB,\widehat{G}_\SB) = o(r_n).
\end{equation*}

\textbf{Technical Details}. The identity \eqref{eq:linearization-Student} can be obtained as follows:
\begin{align*}
    \frac{\Scale}{\widehat{\Scale}} &= 1 - \frac{\widehat{\Scale} - \Scale}{\Scale} \frac{\widehat{\Scale} + \Scale}{\widehat{\Scale} + \Scale} + \frac{(\widehat{\Scale} - \Scale)^2}{\Scale\widehat{\Scale}} \frac{(\widehat{\Scale} + \Scale)^2}{(\widehat{\Scale} + \Scale)^2} \\
    &=  1 - \frac{\widehat{\Scale}^2 - \Scale^2}{2\Scale^2} + \frac{\widehat{\Scale} - \Scale}{2\Scale^2 (\widehat{\Scale} + \Scale)} (\widehat{\Scale}^2 - \Scale^2) + \frac{(\widehat{\Scale}^2 - \Scale^2)^2}{\Scale\widehat{\Scale}(\widehat{\Scale} + \Scale)^2} \\
    &= 1 - \frac{\widehat{\Scale}^2 - \Scale^2}{2\Scale^2} + \frac{(\widehat{\Scale}+2\Scale)(\widehat{\Scale}^2 - \Scale^2)^2}{2\Scale^2 \widehat{\Scale}(\widehat{\Scale} + \Scale)^2}.
\end{align*}

Letting $u_{ij} := \Lincom'U_{ij}$ and defining
\begin{equation*}
    U := \binom{n}{2}^{-1}\sum_{1 \leq i < j \leq n} u_{ij} = \widehat{\theta}_\Lincom, \qquad W_1 := \binom{n}{2}^{-1}\sum_{1 \leq i < j \leq n} u_{ij}^2,
\end{equation*}
and
\begin{equation*}
    W_2 := \binom{n}{3}^{-1}\sum_{1 \leq i < j < k \leq n} \frac{u_{ij}u_{ik} + u_{ij}u_{jk} + u_{ik}u_{jk}}{3},
\end{equation*}
if follows from \citet[Supplemental Appendix]{Cattaneo-Crump-Jansson_2014b_ET} that
\begin{equation*}
    \widehat{\Scale}_{\SB,\Lincom}^2 = \binom{n}{2}^{-1} W_1 + \frac{4}{n}\frac{n-2}{n-1} W_2 - \frac{4}{n}U^2.
\end{equation*}
Also, for $k \in \{2,3\}$, using \citet[Supplemental Appendix]{Cattaneo-Crump-Jansson_2014b_ET} and Lemma \ref{L:Q_Lp},
\begin{align*}
    \E[|U - \E[u_{12}]|^{2k}] &\ls n^{-k} + n^{-2k}h^{-(2p-1)d-2k}, \\
    \E[|W_1 - \E[u_{12}^2]|^k] &\ls n^{-k/2}h^{-k(d+2)} + n^{-k}h^{-(2k-1)d-2p}, \\
    \E[|W_2 - \E[(\E[u_{12}|Z_1])^2]|^k] &\ls n^{-k/2} + n^{-k} h^{-(k-1)d-2p} + n^{-3k/2} h^{-2(k-1)d-2k},
\end{align*}
and therefore
\begin{align}\label{eq:LpError-VarianceEstimator}
    \E\left[|\widehat{\Scale}_{\SB,\Lincom}^2 - \Scale_{\SB,\Lincom}^2|^k \right] &\ls n^{-2k} \E[|W_1 - \E[u_{12}^2]|^k] + n^{-k} \E[|W_2 - \E[(\E[u_{12}|Z_1])^2]|^k] \nonumber \\
    &\qquad + n^{-k} \sqrt{\E[|U - \E[u_{12}]|^{2k}]} \nonumber \\
    &\ls n^{-3k/2} + n^{-2k}h^{-(k-1/2)d-2k} + n^{-5k/2} h^{-2(k-1)d-2k},
\end{align}
where
\begin{equation*}
    \Scale_{\SB,\Lincom}^2 := \binom{n}{2}^{-1} \E[u_{12}^2] + \frac{4}{n} \E[(\E[u_{12}|Z_1])^2] - \frac{4}{n} \E[u_{12}]^2.
\end{equation*}

To prove \eqref{eq:Remainder-Student}, it suffices to show that $\m{V}_1 + \m{V}_2 + \m{V}_3 = o(r_n)$, where
\begin{align*}
    \m{V}_1 &:= \P\left[\frac{(\widehat{\Scale}_{\SB,\Lincom}+2\Scale_{\SB,\Lincom})(\widehat{\Scale}_{\SB,\Lincom}^2 - \Scale_{\SB,\Lincom}^2)^2}{\Scale_{\SB,\Lincom}^2 \widehat{\Scale}_{\SB,\Lincom}(\widehat{\Scale}_{\SB,\Lincom} + \Scale_{\SB,\Lincom})^2} > \frac{r_n}{\log^2 n} \right], \\
    \m{V}_2 &:= \P\left[\frac{|\widehat{\theta}_\Lincom -\theta_\Lincom|}{\Scale_{\SB,\Lincom}} > \log n \right], \\
    \m{V}_3 &:= \P\left[\frac{| \widehat{\Scale}_{\SB,\Lincom}^2 - \Scale_{\SB,\Lincom}^2 |}{\Scale_{\SB,\Lincom}^2} \frac{|B|}{\Scale_{\SB,\Lincom}} > \frac{\sqrt{n}h^P}{\log n} \right].
\end{align*}
First, using \eqref{eq:LpError-VarianceEstimator} and the Markov inequality,
\begin{align*}
    \m{V}_1 &\leq \P\left[(\widehat{\Scale}_{\SB,\Lincom}^2 - \Scale_{\SB,\Lincom}^2)^2  > \frac{r_n \sigma_\Lincom^4}{n^2 \log^2 n} \right] + o(r_n) \ls \frac{n^3 \log^3 n}{r_n^{3/2}} \E\Big[ |\widehat{\Scale}_{\SB,\Lincom}^2 - \Scale_{\SB,\Lincom}^2|^3 \Big] + o(r_n) \\
    &\ls \frac{\log^3 n}{r_n^{3/2}} (n^{-3/2} + n^{-3}h^{-5d/2-6} + n^{-9/2} h^{-4d-6}) + o(r_n) = o(r_n).
\end{align*}
Next, using Theorem \ref{thm:EE-standard} and the properties of the standard normal distribution,
\begin{align*}
    \m{V}_2 &= 1 - \P\left[\frac{\widehat{\theta}_\Lincom -\theta_\Lincom}{\Scale_{\SB,\Lincom}} \leq \log n \right] + \P\left[\frac{\widehat{\theta}_\Lincom -\theta_\Lincom} {\Scale_{\SB,\Lincom}} < - \log n \right] \\
    &= 1 - \Phi(\log n) + \Phi(-\log n) + o(r_n) = o(r_n).
\end{align*}
Finally, using \eqref{eq:LpError-VarianceEstimator} and the Markov inequality,
\begin{align*}
    \m{V}_3 &\ls n \E\Big[ |\widehat{\Scale}_{\SB,\Lincom}^2 - \Scale_{\SB,\Lincom}^2|^2 \Big] \log^2 n  \\
    &\ls (n^{-2} + n^{-3}h^{-3d/2-4} + n^{-4} h^{-2d-4}) \log^2 n = o(r_n).
\end{align*}

Next, to prove \eqref{eq:EE-Student} we begin by using a ``smoothing inequality'' to obtain the bound 
\begin{equation*}
    \rho(\widetilde{F}_\SB,\widehat{G}_\SB) \ls \widehat{\m{I}}_1 + \widehat{\m{I}}_2 + \widehat{\m{I}}_3 + \widehat{\m{I}}_4 + \frac{1}{\sqrt{n}\log n},
\end{equation*}
where
\begin{align*}
    \widehat{\m{I}}_1 &:= \int_{|t| \leq \log n} \left| \frac{\chi_{\widetilde{F}_\SB}(t) - \chi_{\widehat{G}_\SB}(t)}{t} \right| \d t, \qquad \widehat{\m{I}}_2 :=\int_{\log n < |t| \leq c \sqrt{n}} \left| \frac{\chi_{\widetilde{F}_\SB}(t)}{t} \right| \d t, \\
    \widehat{\m{I}}_3 &:=\int_{c \sqrt{n} < |t| \leq \sqrt{n} \log n} \left|\frac{\chi_{\widetilde{F}_\SB}(t)}{t} \right| \d t, \qquad \text{and} \qquad \widehat{\m{I}}_4 :=\int_{|t| > \log n} \left| \frac{\chi_{\widehat{G}_\SB}(t)}{t} \right| \d t.
\end{align*}
Proceeding as in the proof of Theorem \ref{thm:EE-ustat}, it can be shown that $\widehat{\m{I}}_2 + \widehat{\m{I}}_3 + \widehat{\m{I}}_4 = o(r_n)$. The proof of \eqref{eq:EE-Student} can therefore be completed by showing that $\widehat{\m{I}}_1 = o(r_n)$. We shall do so by adapting the proof of Theorem \ref{thm:EE-ustat} to also account for the contribution of $\widehat{\Scale}_{\SB,\Lincom}^2$ to $\widetilde{F}_\SB$.

Defining
\begin{equation*}
    V_\SB := \frac{1}{2n\vartheta_{\SB,\Lincom}^2} n^{-1} \sum_{1 \leq i \leq n} \left( \ell_i^2 - \sigma_\ell^2 + 4 \E[\ell_j q_{ij} | Z_i] \right) + \frac{2}{n\vartheta_{\SB,\Lincom}^2} \binom{n}{2}^{-1} \sum_{1 \leq i < j \leq n} \E[q_{ik}q_{jk}|Z_i,Z_j]
\end{equation*}
and using \cite{Callaert-Veraverbeke_1981_AOS} and the H\"older inequality, it can be shown that
\begin{equation}\label{E:expansion_student}
    \chi_{\widetilde{F}_\SB}(t) = \E \left[ \exp \big( \iota t \widetilde{T}_\SB \big) \right] = \chi_{F}(t) - \iota t \E \left[ \exp \left( \iota t \frac{L}{\Scale_{\SB,\Lincom}} \right) V_\SB \frac{L}{\Scale_{\SB,\Lincom}} \right] + O(\m{T}_1(t)),
\end{equation}
where $\chi_{F}$ was (defined and) analyzed in the proof of Theorem \ref{thm:EE-ustat} and where
\begin{equation*}
    \m{T}_1(t) := \frac{|t|}{n h^{d/2+1} + n^{3/2} h^{3d/2+3}} + \frac{t^2}{n h^{d/2+1} + n^{3/2} h^{3d/2+3}}.
\end{equation*}

As in \citet{Nishiyama-Robinson_2001_ChBook}, the second term on the right-hand side of \eqref{E:expansion_student} admits an expansion of the form
\begin{align}\label{E:extra_term_expasion}
    &\E \left[ \exp \left( \iota t \frac{L}{\Scale_{\SB,\Lincom}} \right) V_\SB \frac{L}{\Scale_{\SB,\Lincom}} \right] \nonumber \\
    &= \chi_\ell(\tau)^{n-1} \frac{\varkappa_1 + 4\varkappa_2}{2 n^2 \Scale_{\SB,\Lincom}^3} - \chi_\ell(\tau)^{n-2} \frac{t^2}{2} \frac{\varkappa_1 + 4\varkappa_2}{n^2 \Scale_{\SB,\Lincom}^3} + \chi_\ell(\tau)^{n-3} O(\m{T}_2(t)), \qquad \tau := \frac{t}{n\Scale_{\SB,\Lincom}},
\end{align}
where $\chi_\ell(\tau)^{n-k}$ was (defined and) analyzed in the proof of Theorem \ref{thm:EE-ustat} and where
\begin{equation*}
    \m{T}_2(t) := \frac{|t|}{n + h^{d+2-P}} + \frac{t^2}{n + n^{3/2}h^{d+2}} + \frac{|t|^3}{n} + \frac{t^4}{n^{3/2} + n^{2}h^{d+2}} + \frac{|t|^5}{n^{5/2}h^{d+2}} + \frac{t^6}{n^{3}h^{d+2}}.
\end{equation*}
Combining \eqref{E:expansion_student} and \eqref{E:extra_term_expasion} with the previously obtained expansions for $\chi_{F}$ and $\chi_\ell(\tau)^{n-k}$, we obtain an expansion of the form
\begin{equation*}
    \chi_{\widetilde{F}_\SB}(t) = \chi_{\widehat{G}_\SB}(t) + \m{F}(t),
\end{equation*}
where
\begin{equation*}
    \widehat{\m{I}}_1 := \int_{|t| \leq \log n} \left| \frac{\m{F}(t)}{t} \right| \d t = o(r_n).
\end{equation*}
\qed

\section{Auxiliary Lemma}

\begin{lem}\label{L:Q_Lp}
    If $\m{S}\subseteq \{(i,j): 1 \leq i,j \leq n\}$ and if $p \geq 2$, then
    \begin{equation*}
        \E\left[\left|\sum_{\{i,j\} \in \m{S}} q_{ij}\right|^p\right]
        \ls \left(\mathfrak{s}_1\E[q_{12}^2]\right)^{p/2} + \mathfrak{s}_{p/2}\E\left[(\E [q_{12}^2|Z_1])^{p/2}\right] + \mathfrak{s}_1\E[|q_{12}|^p],
    \end{equation*}
    where
    \begin{equation*}
        \mathfrak{s}_s:= \max\left[\sum_{1 \leq i \leq n} \left(\sum_{1 \leq j \leq n} \1(\{i,j\}\in\m{S})\right)^s , \sum_{1 \leq i \leq n} \left( \sum_{1 \leq j \leq n} \1(\{j,i\}\in\m{S})\right)^s\right].
    \end{equation*}
\end{lem}

\begin{proof}
    By \citet[Proposition 2.4]{Gine-Latala-Zinn_2000_ChBook}, the inequality holds for the \emph{decoupled} version of $q_{ij}$, defined as $\widetilde{q}_{ij}:=q(Z_i^{(1)},Z_j^{(2)})$ where $\{Z_i^{(k)}:1\leq i\leq  n$, $1\leq k\leq 2\}$ are $i.i.d.$ copies of $Z$. Finally, we can apply the decoupling inequalities in \citet{delaPena-MontgomerySmith_1995_AoP} to obtain the result at the expense of increasing the constant without altering the order of the upper bound; for further details, see \citet[Section 2.5]{Gine-Latala-Zinn_2000_ChBook}.
\end{proof}

%

\bibliography{CFJM_2024_JOE--bib}
\bibliographystyle{jasa}

\end{document}